\documentclass[%
 reprint,
superscriptaddress,
amsmath,amssymb,
aps,
pra,
floatfix,
]{revtex4-1}

\usepackage[T1]{fontenc}
\usepackage{amsfonts}
\usepackage{amsmath}
\usepackage{amssymb}
\usepackage{amsthm}
\usepackage{bm}
\usepackage{mathtools}
\usepackage{dsfont}
\usepackage[dvipsnames]{xcolor}
\usepackage{enumitem}
\usepackage{graphicx}
\usepackage{array}
\usepackage{caption}
\usepackage{subcaption}
\usepackage{multirow}
\usepackage{diagbox}
\usepackage{float}

\usepackage{cellspace}
\usepackage{microtype}
\usepackage{tikz}
\usetikzlibrary{quantikz2}

\theoremstyle{plain}
\newtheorem{thm}{Theorem}
\newtheorem{lem}{Lemma}

\newcommand{\ii}{\mathrm{i}}

\newcommand{\tr}{\mathrm{Tr}}
\newcommand{\fld}{\mathbb{Z}_{d}}
\newcommand{\lhv}{\mathrm{lhv}}

\usepackage{natbib}
\usepackage{hyperref}

\begin{document}
\preprint{APS/123-QED}

\title{Robustly self-testing all maximally entangled states in every finite dimension}

\author{Uta Isabella Meyer}
\email{uta-isabella.meyer@lip6.fr}
\affiliation{Sorbonne Université, CNRS, LIP6, F-75005 Paris, France}
\author{Ivan Šupić}
%\email{ivan.supic@univ-grenoble-alpes.fr}
\affiliation{Université Grenoble Alpes, CNRS, Grenoble INP, LIG, 38000 Grenoble, France}
\author{Frédéric Grosshans}
%\email{frederic.grosshans@lip6.fr}
\affiliation{Sorbonne Université, CNRS, LIP6, F-75005 Paris, France}
\author{Damian Markham}
%\email{damian.markham@lip6.fr}
\affiliation{Sorbonne Université, CNRS, LIP6, F-75005 Paris, France}

\date{\today}

\begin{abstract}
We establish a device-independent, noise-tolerant certification of maximally entangled states in every finite dimension $d$.
The core ingredient is a $d$-input, $d$-outcome Bell experiment that generalizes the Clauser–Horne–Shimony–Holt test from qubits to qudits where each setting is a non-diagonal Heisenberg–Weyl observable.
For every odd prime $d\geq 3$, the associated Bell operator has an exact sum-of-positive-operators decomposition, yielding the Cirelson bound in closed form, from which we reconstruct the Heisenberg–Weyl commutation relations on the support of the state.
We then extend the Mayers-Yao local isometry from qubits to prime-dimensional systems and show that any $\epsilon$-near-optimal strategy below that bound is, up to local isometries, within trace distance $\delta = \mathcal{O}(\sqrt{\epsilon })$ of the ideal maximally entangled state; the implemented measurements are correspondingly close to the target observables.
Via a tensor-factor argument, the prime-dimension result extends the self-testing protocol to every composite dimension $d$.
The protocol uses standard Heisenberg–Weyl operations and non-Clifford phase gates that are diagonal in the computational basis, making it directly applicable to high-dimensional photonic and atomic platforms.
\end{abstract}
\maketitle

\paragraph*{Introduction.}

Bell's theorem reshaped our understanding of nature by demonstrating that quantum predictions fundamentally diverge from classical expectations, ruling out any explanation via local-hidden-variable (LHV) models \cite{bell1964einstein,RevModPhys.38.447}.
Today, Bell experiments not only serve as powerful tools to demonstrate nonlocality but also form the foundation for device-independent self-testing, enabling certification of quantum states and measurements purely from observed correlations \cite{mayers2003self,Supic2020selftestingof}.
Self-testing has emerged as a pivotal technique in quantum information science, with crucial implications for secure quantum communication~\cite{Pironio2010,Colbeck_2011,Arnon-Friedman2019}, delegated quantum computation~\cite{Reichardt2013} and computational complexity~\cite{ji2021mip}.

While self-testing is well-developed and robustly characterized for qubit systems, extending these certification protocols to higher-dimensional quantum systems (qudits) remains an open challenge. Qudits enable stronger dimension witnesses and richer forms of genuine multipartite entanglement \cite{GUHNE20091,PhysRevLett.100.210503,PhysRevLett.106.250404,brunner2014bell}, exhibit novel nonlocal features \cite{Collins2001,Kaszlikowski2002,Kaszlikowski2002-2,Acin2004,Ji2008,Liang2009,Gruca2011,Tang2013,Howard2015bell}, and support a variety of cryptographic and communication protocols \cite{Fitzi2001,Hadley2005,shaari2008blind,Keet2010,PhysRevA.78.012353,PhysRevLett.102.120501,Mackeprang2023}.
Experimentally, high-quality qudits have been realized using photonic path, time-bin and orbital-angular-momentum modes \cite{PhysRevA.55.2564,doi:10.1126/science.aar7053}.
Nonetheless, despite this promise, existing self-testing protocols for higher-dimensional systems either rely on decompositions into qubit subspaces~\cite{PhysRevA.87.050102,Coladangelo2017,PhysRevA.98.052115} or lack complete robustness analyses~\cite{Kaniewski2018,Cui2020generalizationof,Saha2020sumofsquares,Sarkar2021}.

This paper addresses this critical gap by introducing a robust device-independent self-testing protocol that certifies maximally entangled states across all finite dimensions.
We achieve this by generalizing the well-known Clauser–Horne–Shimony–Holt (CHSH)~\cite{chsh1969} test from qubits to prime-dimensional qudits using non-diagonal Heisenberg–Weyl observables, and building upon recent proposals introduced by the authors in previous work~\cite{meyer2024bell}.
By a tensor-factor argument, the protocol extends to every composite dimension.
We explicitly derive analytic robustness bounds, showing that correlations approaching the quantum Cirelson bound~\cite{Cirelson1980} enforce strict approximations of the ideal maximally entangled state and canonical Heisenberg–Weyl measurement operators.

\paragraph*{Previous works.}

Early works constructed tailored Bell inequalities or extended CHSH-type tests for higher dimensions~\cite{Coladangelo2017}, but these often required multiple correlation constraints or dimension-specific designs.
A breakthrough by Sarkar et al.~showed that the maximally entangled state in an arbitrary local dimension $d$ can in principle be certified using a single Bell scenario with two $d$-outcome measurements per party~\cite{Sarkar2021}, albeit without a rigorous robustness analysis.
A similar self-testing of maximally entangled states of local dimension $2^d$ for every $d \geq 1$, with two measurements per party has been reported in~\cite{Supic2021deviceindependent}, but robustness analysis lacks as well.
In fact, the constructions used in both of these works cannot be made robust by standard approach using triangle and Cauchy-Schwarz inequalities, as noted in~\cite{baroni2024quantumboundscompiledxor}.
Subsequent protocols achieved constant-sized self-tests in which the number of measurements does not grow with the number of inputs: Fu exhibited explicit fixed-size correlations self-testing an infinite family of maximally entangled states with unbounded local dimension via number-theoretic constructions~\cite{Fu2022constantsized}, and Man\v{c}inska et al.~developed a fully robust algebraic scheme that certifies every maximally entangled state of odd dimension (and even arbitrarily high-rank two-outcome measurements) using only four binary measurement settings~\cite{manvcinska2024constant}.
More recently, Vol\v{c}i\v{c} lifted the parity restriction by proving that four binary measurements per party suffice to self-test every maximally entangled state, extending the constant-size approach to all finite dimensions~\cite{volvcivc2024constant}.
In parallel, Chen et al.~introduced “post-hoc” iterative self-testing to certify additional observables, enabling device-independent verification of essentially any real projective measurement on the self-tested high-dimensional maximally entangled state~\cite{chen2024all}.
Each of these advances comes with trade-offs, some require tailored inequalities for each $d$ or number-theoretic conditions~\cite{Fu2022constantsized}, others were limited to either even or odd local dimensions~\cite{manvcinska2024constant}, or relied on idealized assumptions and lacked explicit error-tolerance, highlighting the challenge of a universal, robust high-dimensional self-test.
Our result overcomes these limitations, providing the first single protocol that is both fully robust and universally applicable to all local dimensions $d$.
In contrast to the earlier two-setting scheme of Sarkar et al.~\cite{Sarkar2021} which lacked a robustness guarantee, this work includes a rigorous error-bound analysis ensuring certification under experimental noise.
Unlike in our approach, one of the main challenges in developing robust self-tests for maximally entangled qudit states in~\cite{Fu2022constantsized,manvcinska2024constant,volvcivc2024constant} is identifying a universal, analytic form of measurements that achieve maximal violation across arbitrary dimensions.
The absence of such explicit analytic expressions complicates practical implementations of these self-tests, necessitating dimension-specific constructions or numerical optimization techniques.

\paragraph*{Nonlocality \& self-testing.}

A bipartite Bell experiment involves two parties, each receiving an input ($x$ and $y$) and producing an output ($a$ or $b$) from finite sets. After many repetitions, this process leads to an empirical distribution $p(a,b|x,y)$. We assume the no‑signalling condition, meaning that the marginal distribution of one party is independent of the other’s input $\sum_a p(a,b\vert x,y) = \sum_a p(a,b\vert x',y) $ and $\sum_b p(a,b\vert x,y) = \sum_b p(a,b\vert x',y) $ for all $a(b),x,x',y,y'$.
Nonlocality refers to the impossibility of reproducing such a distribution with any LHV model of the form $p(a,b\vert x,y) = \sum_{\lambda} p(\lambda) p(a\vert x,\lambda) p(b\vert y,\lambda)$.
Certain correlations arising from local measurements on entangled quantum states violate this classical description and are therefore nonlocal~\cite{GISIN1991201}.
The distinction from classical models is typically certified by a Bell inequality, i.e.~a bound on the expectation value of a Bell operator under LHV models.

Some quantum correlations go further: they are not only nonlocal but also unique, in the sense that no alternative Bell experiment can reproduce them, except up to irrelevant degrees of freedom such as local unitaries or actions on auxiliary subsystems.
When all experimental realizations of the observed correlations can be shown to be equivalent in this sense to a known reference experiment, one can conclude that the underlying state and measurements are uniquely characterized.
This property is called \textit{self-testing}~\cite{Supic2020selftestingof}. \\[0.1cm]

\paragraph*{Qudits.}

We consider finite dimensional quantum systems called qu$d$its, here, for odd prime dimensions $d$.
Qudits can be seen as a generalization of qubits $d=2$ to higher dimensions.
Their fundamental observables extend the qubit Pauli operators: the generalized shift operator $X$ and phase operator $Z$ act as $X \vert k \rangle = \vert (k+1)\,\mathrm{mod}\,d\, \rangle$ and $Z \vert k\rangle = \omega^k \vert k \rangle$ where $\omega = \exp (2 \pi \ii / d)$ is a primitive $d^{\text{th}}$ root of unity.
These operators are unitary with $X^{-1} = X^{d-1} = X^{\dagger}$ and $Z^{-1} = Z^{d-1} = Z^{\dagger}$ and they satisfy the twisted commutation relations \( Z^{z} X^{x} = \omega^{x z} X^{x} Z^{z}\) for $x,z \in \fld$.
In the following, all integers are elements of and all arithmetic are considered over the finite field $\fld$.
From $X$ and $Z$, one constructs the Heisenberg-Weyl (HW) displacement operators
\begin{equation} T_{(x,z)} = \omega^{2^{-1} xz} X^{x} Z^{z} \label{eq:displace} \end{equation}
are the generalized qudit Pauli operators.
Every qudit state $\rho$ has a unique representation in terms of its overlap with the HW displacement operators.
The so-called characteristic function is the Fourier transform of Gross' Wigner function~\cite{Gross2006} and is defined as
\begin{equation}
\chi[\rho]_{(u_{x},u_{z})} = \tr \left(T_{(u_{x},u_{z})}^{\dagger} \rho \right)\,.
\label{eq:qudit_characteristic}
\end{equation}

\paragraph*{Previous Bell experiment.}

We reintroduce our novel generalization of the well-known CHSH Bell test presented in~\cite{meyer2024bell}.
A quantum realization of the experiment consists of measurements in the eigenbases of the Heisenberg–Weyl (HW) displacement operators $T_{u}$ indexed by $u \in \fld^{2}$.
Specifically, each input specifies one of the mutually unbiased bases associated with the subgroups
\begin{equation}
\begin{aligned}
    \mathcal{G}_{k} &= \{ T_{(n,nk)},\,n\in \fld \}\,,~ k \in \fld \,,\\ \mathcal{G}_{d} &= \{ T_{(0,n)},\,n\in \fld \} \,.
\end{aligned}
\end{equation}
The initial state of the Bell experiment is the maximally entangled qudit Bell state locally rotated by a non-Clifford unitary $U$,
\begin{equation}
\vert \Phi \rangle = \frac{1}{\sqrt{d}} (U \otimes \mathds{1}) \sum_{k \in \fld} \vert k k \rangle \,.
\label{eq:bellstate}
\end{equation}
he presence of $U$ is essential, as it introduces the non‑classical features of the experiment~\cite{Spekkens2008,Howard2014,Delfosse2017}. Operationally, one may equivalently view this setup as performing HW measurements in rotated bases on the standard Bell state.

We choose to use a specific family of unitary operators
\begin{equation}
U_{\nu} = \sum_{k\in \fld} \omega^{\nu_{k}} \vert k \rangle \langle k \vert \,, \label{eq:magicunitary}
\end{equation}
parameterized by functions $\nu$ that generalize the qubit $\pi/8$-gat, introduced  by Howard and Vala~\cite{Howard2012}.
For $d>3$, the function $\nu$ is a polynomial of at least degree $\mathrm{deg}(\nu) = 3$.
With these choices, the Bell operator is
\begin{equation}
\mathcal{B}_{d}^{(\nu)} = \sum_{u,v \in \fld^{2}} \chi[ \left( U_{\nu} \otimes \mathds{1} \right) \vert \Phi \rangle ]_{(u,v)} \, T_{u} \otimes T_{v} \,,\label{eq:belloperatornu}
\end{equation}
where $\chi$ denotes the characteristic function. In Appendix~\ref{sec:character5}, we calculate explicitly that
\begin{equation}
\begin{aligned}
    &\chi[ \left( U \otimes \mathds{1} \right) \vert \Phi \rangle]_{(x_{1},z_{1}),(x_{2},z_{2})} \\ 
    &= \frac{1}{d}\, \delta_{x_{1} = x_{2}} \sum_{s \in \fld} \omega^{ (z_{1} + z_{2})s + \nu_{s+ 2^{-1}x_{1}} -\nu_{s- 2^{-1}x_{1}}} \,.
\end{aligned}
\end{equation}
Although our analysis allows general functions $\nu$, Howard and Vala identified a minimal family where $\nu$ is a degree-three polynomial~\cite{Howard2012}, and this specual case is detailed in Appendix\:\ref{sec:decomp}.
  
The expectation value of the Bell operator on the maximally entangled state is $\langle \Phi \vert \mathcal{B}_{d}^{(\nu)} \vert \Phi \rangle = d^{2}$.
In~\cite{meyer2024bell}, we conjectured, and verified numerically for $5 \leq d \leq 23$, that any LHV model satisfies $\langle \mathcal{B}_{d}^{(\nu)} \rangle_{\lhv} < d^2 = \langle \Phi \vert \mathcal{B}_{d}^{(\nu)} \vert \Phi \rangle $ for all $d$ odd prime~\cite{meyer2024bell}.\\[0.1cm]

\paragraph*{Present Bell experiment.}

In the present work, we consider a refinement of the above Bell experiment, obtained by omitting the measurements in the $Z$ basis, i.e. those associated with the subgroup $\mathcal{G}_{d}$.
This leaves the scenario with $d$ inputs and $d$ outputs.
The reason is that the non-Clifford unitaries~\eqref{eq:magicunitary} - essential for non-classical correlations - are diagonal in the $Z$ basis~\cite{Howard2012}.
Retaining this setting would therefore restrict the state to a stabilizer subspace, since $Z \otimes Z^{\dagger}$ is still a stabilizer element of the state~\eqref{eq:bellstate}.
Such measurements admit perfect classical simulation~\cite{Howard2013}, and thus contribute nothing to certifying nonlocality or enabling self‑testing.

To accommodate the missing basis, we rewrite $T_{(0,z)}$ as $\omega^{-2^{-1}z} T_{(-1,0)} T_{(1,z)}$.
The resulting Bell operator $\mathcal{B}_{d}$ is then composed exclusively of non‑diagonal terms.
In fact, the Bell operator~\eqref{eq:belloperatornu} can be written as a sum of three independent terms $\mathcal{B}_{d}^{(\nu)} = \mathds{1} + \mathcal{S} + \mathcal{B}_{d}$.
The purely diagonal (non-trivial) elements are combined to
\begin{equation}
\mathcal{S} = \sum_{j \in \fld^{\ast}} A_{0,j} \otimes B_{0,-j} \,. \label{eq:belloperatorstabself}
\end{equation}
and the remaining non‑diagonal contributions form
\begin{equation}
\mathcal{B}_{d} = \sum_{n \in \fld^{\ast}} \sum_{j,k \in \fld} g(j,k,n) \,A_{j}^{n} \otimes B_{k}^{n} \,,
\label{eq:belloperator5self}
\end{equation}
with
\begin{equation}
g(j,k,n) = \frac{1}{d} \sum_{s \in \fld} \omega^{ n(j + k) s + \nu_{s+ 2^{-1}n} -\nu_{s- 2^{-1}n}} \,.
\label{eq:gdefinition}
\end{equation}
For the state\,\eqref{eq:bellstate}, the operator's expectation value is $\langle \Phi \vert \mathcal{B}_{d} \vert \Phi \rangle = d(d-1)$.
A key property for our robustness analysis is that $d(d-1) \mathds{1} - \mathcal{B}_{d}$ admits an explicit sum-of-positive-operators (SOPO) decomposition:
\begin{equation}
d(d-1) \mathds{1} - \mathcal{B}_{d} = \sum_{n = 1}^{\tfrac{d-1}{2}} \sum_{j \in \fld} C_{n,j}^{\dagger} C_{n,j}^{\phantom{\dagger}} \,,
\label{eq:decompA}
\end{equation}
where 
\begin{equation}
C_{n,j} := A_{j}^{n} - \sum_{ k \in \fld } \,g(j,k,n)^{\ast}\, \left( B_{k}^{n} \right)^{\dagger} \,,
\end{equation}
and $C_{n,j}^{\dagger} = C_{-n,j}^{\phantom{\dagger}}$.
This decomposition is verified explicitly in Appendix\:\ref{sec:decomp}.

The case $d=3$ requires a separate treatment.
Both the choice of $\nu$ and the form of the Bell operator differ, requiring a more delicate analysis to establish the self‑test. Specifically, one considers the operator
\begin{equation}
\label{eq:BellOperator}
\begin{aligned} \mathcal{B}_{3} &=\frac{1}{\sqrt{3}}\,\big(e^{\ii\phi_{1}} A_{0} \otimes B_{0} + e^{\ii\phi_{2}} A_{0} \otimes B_{1} + e^{\ii\phi_{1}} A_{0} \otimes B_{2}\\[-0.25cm]
&\phantom{\frac{1}{\sqrt{3}}}+ e^{\ii\phi_{2}} A_{1} \otimes B_{0} + e^{\ii\phi_{1}} A_{1} \otimes B_{1} + e^{\ii\phi_{1}} A_{1} \otimes B_{2}\\[-0.25cm]
&\phantom{\frac{1}{\sqrt{3}}}+ e^{\ii\phi_{1}} A_{2} \otimes B_{0} + e^{\ii\phi_{1}} A_{2} \otimes B_{1} + e^{\ii\phi_{2}} A_{2} \otimes B_{2}
\big)+h.c. \, \end{aligned} 
\end{equation}
where the phases $\phi_{1},\phi_{2}$ are constrained by
\begin{equation}
\phi_{1}-\phi_{2} \equiv \frac{2\pi}{3}\,, \qquad 1 - \sqrt{3} e^{3 \ii \phi_{1}} = \omega \,.
\label{eq:phaseDifference}
\end{equation}
Note that the Bell operator\,\eqref{eq:BellOperator} can be written in the form of Eq.\,\eqref{eq:belloperator5self}, and there exists a Bell operator $B_{\nu} = \mathds{1} + \mathcal{S}_{3} + \mathcal{B}_{3}$ that has the form of Eq.\,\eqref{eq:belloperatornu}.
The operator $\mathcal{S}_{3} = A_{3}^{\phantom{\dagger}} \otimes B_{3}^{\phantom{\dagger}} + A_{3}^{\dagger} \otimes B_{3}^{\dagger} = 2\mathds{1} - (A_{3}^{\dagger} - B_{3}^{\phantom{\dagger}})(A_{3}^{\phantom{\dagger}} - B_{3}^{\dagger})\,,\label{eq:BellOperatorStab3}$ is again diagonal in the computation basis and part of the present Bell experiment.
Unlike for the case $d > 3$, $\nu$ cannot be a polynomial in the finite field.
Instead, a possible choice is $\phi_{1} = -\pi/18$ and $\phi_{2} = -13\pi/18$, for which $\langle \mathcal{B}_{3} \rangle_{\lhv} < 5.640 $.
Given the qutrit Bell state, it is $\langle \Phi \vert \mathcal{B}_{3} \vert \Phi \rangle = 6 $.

Moreover, there exist a SOPO decomposition $6 \mathds{1} - \mathcal{B}_{3} = C_{0}^{\dagger} C_{0}^{\phantom{\dagger}} + C_{1}^{\dagger} C_{1}^{\phantom{\dagger}} + C_{2}^{\dagger} C_{2}^{\phantom{\dagger}}$ with
\begin{equation}
C_{j} := A_{j} - \frac{1}{\sqrt{3}} \sum_{k=0,1,2} e^{-\ii \phi_{j,k}} B_{k}^{\dagger} \,, \label{eq:sos3-j}
\end{equation}
where $\phi_{0,1} = \phi_{1,0} = \phi_{2,2} = \phi_{2}$, and otherwise $\phi_{j,k}= \phi_{1}$.\\[0.1cm]

\paragraph*{Main result}

We now state and prove our main theorem, which establishes that the maximally entangled qudit state and the canonical Heisenberg–Weyl measurements can be robustly self-tested from the maximal violation of our Bell inequality.
\begin{thm}
\label{th:main-ideal}
Given unitary operators $A_{j} \in \mathcal{H}_A$, $B_{j} \in \mathcal{H}_B$ for $j=0,\dots,d-1$ that constitute the operator $\mathcal{B}_{d}$ in Eq.\,\eqref{eq:belloperator5self} when $d > 3$, and the operator $\mathcal{B}_{3}$ in Eq.\,\eqref{eq:BellOperator} if $d=3$.
Suppose futher that $\vert \psi \rangle \in \mathcal{H}_A \otimes \mathcal{H}_B$ is a bipartite state such that $\vert \langle \psi \vert \mathcal{B}_{d} \vert \psi \rangle - d(d-1) \vert \leq \epsilon$.
Then, there exist local unitaries $V_{A} : \mathcal{H}_{A} \rightarrow \mathbb{C}^d \otimes \mathcal{H}_{A'}$ and $V_{B} : \mathcal{H}_{B} \rightarrow \mathbb{C}^d \otimes \mathcal{H}_{B'} $ together with an error function $\delta(\epsilon) = \sqrt{\epsilon} \, d(d-1) (\mu_{d}(4+1/d) + 1)$ such that
\begin{equation}
\Vert \left( V_{A}\otimes V_{B} \right) (\vert \psi \rangle ) - \vert\Phi \rangle \otimes \vert aux \rangle \Vert \leq \delta,
\end{equation}
where $\vert\Phi \rangle$ is the maximally entangled state as in Eq.\,\eqref{eq:bellstate}, $\vert aux \rangle \in \mathcal{H}_{A'} \otimes \mathcal{H}_{B'} $ is an auxiliary state, $\mu_{3} = 9\sqrt{\epsilon} (\sqrt{3}+2)$, and $\mu_{d} = \sqrt{d}\big( \sqrt{d} + 2 \big)$ when $d>3$.
For $d=3$, we further assume that $\epsilon < (72(\sqrt{3}+2))^{-2}$.
Moreover, under these unitaries the observables are mapped as  $$V_{A} (A_{j} ) \rightarrow T_{(1,j)} \otimes \mathds{1}_{A'}, \qquad V_{B} (B_{j}) \rightarrow T_{(1,j)} \otimes \mathds{1}_{B'},$$ where $T_{(1,j)}$ are the canonical Heisenberg-Weyl operators.
The auxiliary spaces $\mathcal{H}_{A'}$ and $\mathcal{H}_{B'}$ may have arbitrary finite dimensions.
\end{thm}
We first present the proof in the ideal case, $\delta = \epsilon = 0$, where the argument is somewhat cleaner.
The robustness, given in Appendix~\ref{sec:robust}, follows the same structure: each algebraic step is reproduced with error bounds, leading to the stated dependence on $\epsilon$.
Both the ideal and robust proofs have a two-step structure:
\begin{enumerate}
    \item \textit{Twisted commutation relations.} We show that the optimal violation forces the measurement operators to satisfy twisted commutation relations, characteristic of the HW algebra (Lemma~\ref{lem:com-ideal} for the ideal case).
    \item \textit{Isometry construction.} Using these relations, we explicitly build a local isometry that extracts the maximally entangled state and maps the physical observables to their canonical HW counterparts. 
\end{enumerate}

\begin{lem}
\label{lem:com-ideal}
Let $A_{j}$, $B_{j}$  and $\ket{\psi}$ be as in Theorem~\ref{th:main-ideal}, and assume 
 $\langle \psi \vert \mathcal{B}_{d} \vert \psi \rangle = d(d-1)$. Then for all$j,k \in \fld$ the following twisted commutation relations hold:
\begin{equation}
\begin{aligned}
    \left( A_{j} A_{k} - \omega^{j-k} A_{k} A_{j} \right) \vert \psi \rangle &= 0\,,\\
    \left( B_{j} B_{k}- \omega^{j-k} B_{k} B_{j} \right) \vert \psi \rangle &= 0\,.
\end{aligned}
\label{eq:lemcomrel}
\end{equation}
\end{lem}
For the proof, we use the SOPO decomposition\,\eqref{eq:decompA} to obtain a relation between the $A_{j}$ and $B_{k}$ with respect to their action on the state $\vert \psi \rangle$.
And then deduce the commutation relation as the only solution to the restrictions $A^{2}_{j} = (A_{j})^{2}$ for all $j \in \fld$.
\begin{proof}[Proof of Lemma~\ref{lem:com-ideal} for $d > 3$]
Assume unitary operators $A_{j}, B_{j}$ that constitute the Bell operator\,\eqref{eq:belloperator5self} and a state $\vert \psi \rangle$ such that $\langle \psi \vert \mathcal{B}_{d} \vert \psi \rangle = d(d-1)$.
This implies $\langle \psi \vert d(d-1) - \mathcal{B}_{d} \vert \psi \rangle = 0$.
Since the SOPO decomposition~\eqref{eq:decompA} expresses this operator as a sum of positive terms, it follows that every term must vanish individually.
In particular, $\langle \psi \vert C_{n,j}^{\dagger} C_{n,j}^{\phantom{\dagger}} \vert \psi \rangle =0$, which implies $C_{n,j}\vert \psi \rangle = 0$ for all $n \in \fld^{\ast}$ and $j \in \fld$. Explicitly,
\begin{equation}
A_{j}^{n} \vert \psi \rangle = \sum_{ k \in \fld } \,g(j,k,n)^{\ast} \left( B_{k}^{n} \right)^{\dagger}\vert \psi \rangle \,.
\label{eq:abrelation}
\end{equation}

It generally holds that $A_{j}^{n} A_{j}^{n'} = A_{j}^{n+n'}$ for all $n,n',j \in \fld$, and, in particular, $A_{j}^{d} = \mathds{1}$, as well as
\begin{equation}
A_{j} A_{j} = A_{j}^{2} \,. \label{eq:aarelation}
\end{equation}
Taking the adjoint of Eq.\,\eqref{eq:aarelation}, we have $\left( A_{j}\right)^{\dagger} \left( A_{j}\right)^{\dagger} = \left( A_{j}^{2} \right)^{\dagger}$.
Combining this with Eq.\,\eqref{eq:abrelation} and Eq.\,\eqref{eq:aarelation} leads to the linear system of equations
\begin{equation}
\sum_{ k,l \in \fld } \, g(j,k,1) g(j,l,1) \, B_{k} B_{l} \vert \psi \rangle = \sum_{ k \in \fld } \, g(j,k,2) B_{k}^{2} \vert \psi \rangle, \label{eq:linsys}
\end{equation}
valid for all $j \in \fld$.
A solution of this system is
\begin{equation}
B_{k} B_{l} \vert \psi \rangle = \omega^{2^{-1} ( k - l)} B_{2^{-1}(k+l)}^{2} \vert \psi \rangle \,,
\end{equation}
as shown in Appendix~\ref{sec:comrel}.

Up to a global phase, this solution is unique.
Expand every grade-$2$ vector in the orthonormal basis $\{B_{m}^{2}\lvert\psi\rangle,\,m\in\fld\}$, $B_{k}B_{l} \vert\psi\rangle = \sum_{m\in\fld} C_{kl}^{m}\,B_{m}^{2} \vert \psi \rangle$, insert this into Eq.\,\eqref{eq:linsys}, and project with $\langle\psi \vert ( B_{n}^{2} )^{\dagger}$.
For all $j,n\in\fld$, one obtains
\begin{equation}
\sum_{k,l \in \fld} C_{kl}^{n} \; g(j,k,1)\, g(j,l,1) = g(j,n,2)\,.\label{eq:linearsystemmatrix}
\end{equation}
Collecting the $C_{kl}^{n}$ in a column vector $\mathbf c^{(n)}$ and the products $g(j,k,1)g(j,l,1)$ in the rows of a $d\times d^{2}$ matrix $M$, Eq.\,\eqref{eq:linearsystemmatrix} reads $M \mathbf{c}^{(n)}=\mathbf g^{(2)}_{n}$.
Because $g(j,k,1)$ is the $(j+k)$-th Fourier coefficient of a non-linear phase ($\nu$ is not a polynomial over the finite field with degree 2 or lower), the rows of $M$ are orthogonal, $\sum_{j\in \fld}g(j,k,1)g(j,l,1)^{\ast}= d \, \delta_{k,l}$, see Appendix\:\ref{sec:comrel}. 
As a result, $M$ has full row rank $d$.
Because the matrix is of maximal rank, the solution is uniquely determined (up to an overall phase), $C_{kl}^{n} =\omega^{2^{-1}(k-l)}\,\delta_{n,2^{-1}(k+l)}$.
Were $\nu$ quadratic or linear, $M$ would drop rank, extra solutions would arise, and $\vert \psi\rangle$ would degenerate into a stabilizer state - contradicting its certified non-classicality under HW group~\cite{Howard2013}.
Finally, it follows that $B_{k} B_{l} \vert \psi \rangle= \omega^{ k - l} B_{l} B_{k} \vert \psi \rangle $ for all $k,l \in \fld$.

To obtain $A_{k} A_{l} \vert \psi \rangle = \omega^{ k - l} A_{l} A_{k} \vert \psi \rangle$ for all $k,l \in \fld$, note that the Bell operator\,\eqref{eq:belloperator5self} is symmetric under swapping the indices of the local operators.
One can also repeat the above using another SOPO decomposition $d(d-1) \mathds{1} - \mathcal{B}_{d} = \sum_{n = 1}^{(d-1)/2} \sum_{j \in \fld} D_{n,j}^{\dagger} D_{n,j}^{\phantom{\dagger}}$ with $ D_{n,j} := B_{j}^{n} - \sum_{ k \in \fld } g(j,k,n) \left( A_{k}^{n} \right)^{\dagger} $.
\end{proof}

For the case $d = 3$, the reasoning in the above proof does not apply because $2 \equiv -1 \,(\bmod \, 3)$, and the relations $A_{j} A_{j} = A_{j}^{2} = A_{j}^{\dagger}$ are not independent.
The linear system\,\eqref{eq:linsys} is therefore under-determined.
Instead, we derive an alternative set of equations involving symmetric operator combinations $\{B_{j},B_{k}\} = B_{j}B_{k} + B_{k} B_{j}$ for $j,k \in \{ 0,1,2 \}$.
\begin{proof}[Proof of Lemma~\ref{lem:com-ideal} for $d = 3$]
Under the assumptions of Lemma~\ref{lem:com-ideal} using the Bell operator\,\eqref{eq:BellOperator}, it is $\langle \psi \vert 6 - \mathcal{B}_{3} \vert \psi \rangle = 0$.
As a consequence, $C_{j} \vert \psi \rangle = 0$, relating the operators $A_{j}$ to the $B_{k}$.
Additionally, it must hold $A_{j}^{\dagger}=A_{j} A_{j}$ for $j \in \{ 0,1,2 \}$.
Writing these constraints out gives a linear system of equations that can be simplified to
\begin{equation}
\begin{aligned}
\left[ \{B_{0}^{2},B_{1}^{2}\} + B_{2} \right] \vert \psi \rangle &= 0 \,,\\
\left[ \{B_{0}^{2},B_{2}^{2}\} + B_{1} \right] \vert \psi \rangle &= 0 \,,\\
\left[ \{B_{1}^{2},B_{2}^{2}\} + B_{0} \right] \vert \psi \rangle &= 0 \,.
\end{aligned}
\label{eq:bconditions}
\end{equation}
In Appendix\:\ref{sec:comrel}, it is shown that the commutation elements $Q:=B_{0}\,B_{1}\,B_{0}^{\dagger}B_{1}^{\dagger}$, $Q^{\prime} = B_{2} B_{0} B_{2}^{\dagger} B_{0}^{\dagger}$, $Q^{\prime\prime} = B_{1} B_{2} B_{1}^{\dagger} B_{2}^{\dagger}$ are central: they commute with  $B_{i}$ for $i = 0,1,2$.
Hence they must be scalar multiples of the identity, $Q^{('('))}=\lambda\mathds{1}$ with $\lambda^{3}=1$ and $\lambda\neq 1$.
Choosing the orientation \(0 \rightarrow 1 \rightarrow 2 \rightarrow 0\) sets \(\lambda=\omega\) (the opposite orientation would give \(\omega^{2}\)).
With \(Q^{('('))} = \omega \mathds{1}\) the defining equation becomes the $\omega$-commutation rules
\begin{equation}
B_{0}B_{1}=\omega B_{1}B_{0},\quad
B_{1}B_{2}=\omega B_{2}B_{1},\quad
B_{2}B_{0}=\omega B_{0}B_{2}\,.
\label{eq:heisenbergweylalgebra}
\end{equation}
which generate the Heisenberg-Weyl algebra in dimension $3$.
Using these $\omega$-commutation relations, we can straightforwardly deduce the same relations for the operators $A_{0}, A_{1}, A_{2}$.
Likewise, it is possible to deduce them independently, but in the exact same fashion as for the $B_{0},B_{1},B_{2}$, using the SOPO decomposition $6 \mathds{1} - \mathcal{B}_{3} = D_{0}^{\dagger} D_{0}^{\phantom{\dagger}} + D_{1}^{\dagger} D_{1}^{\phantom{\dagger}} + D_{2}^{\dagger} D_{2}^{\phantom{\dagger}} $ with $D_{j} := B_{j} - \tfrac{1}{\sqrt{3}} \sum_{k=0,1,2} e^{-\ii \phi_{j,k}} A_{k}^{\dagger}$.
\end{proof}

\begin{proof}[Proof of Theorem~\ref{th:main-ideal}]
Under the assumptions of Theorem~\ref{th:main-ideal}, we use the SOPO~decomposition\,\eqref{eq:decompA} such that Eq.\,\eqref{eq:abrelation} holds and Lemma~\ref{lem:com-ideal} applies.
Then, there exist two families of relations: the twisted commutation relations\,\eqref{eq:lemcomrel}, and those between $A_{j}$ and $B_{j}$ for all $j \in \fld$ in Eq.\,\eqref{eq:abrelation}.
Then, we construct a local isomorphism $V_A \otimes V_B$:
Each $V_A,V_B$ is defined by the circuit in Fig.~\ref{fig:swapcircuit} which accepts the given state and the state $\vert 0 \rangle$ on an auxiliary qudit.
For $X := A_{0}^{\dagger}$, $Z:= \omega^{-2} A^{\dagger}_{0} A_{1} $ for $V_A$, and $X := B_{0}^{\dagger}$, $Z:= \omega^{-2} B^{\dagger}_{0} B_{1} $ for $V_B$.
In Appendix~\ref{sec:iso}, we directly calculate, using the two families of relations, that this isometry maps the state $\vert \psi \rangle$ to the rotated Bell state, and the operators $A_{j}$ and $B_{j}$ to the canonical Heisenberg-Weyl operators $T_{(1,j)}$, up to rotational degrees of freedom and some auxiliary space that the experiment acts on trivially.
\end{proof} 
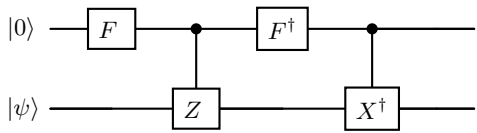
\begin{figure}[ht]
\centering
\begin{quantikz}
\lstick{$\ket{ 0}\,$}& \gate{F\phantom{'}} & \ctrl{1} & \gate{F^{\dagger}} & \ctrl{1} & \qw & \qw \\
\lstick{$\ket{\psi}$}& \qw & \gate{Z\phantom{'}} & \qw & \gate{X^{\dagger}} &\qw&\qw
\end{quantikz}
\caption{Circuit description of a SWAP gate between an input state $\vert \psi \rangle$ and the state $\vert 0 \rangle$ in qudits with the Fourier operation (, a generalized Hadamard matrix), $F = \tfrac{1}{\sqrt{d}} \sum_{i,j \in \fld} \omega^{jk} \vert j \rangle \langle k \vert$, and Heisenberg-Weyl operators $X =T_{(1,0)}, Z= T_{(0,1)} $.}
\label{fig:swapcircuit}
\end{figure}

\paragraph*{Discussion.}

We conclude that robustly self-testing in any prime dimension is possible via a simple self-test that generalizes known ideas from self-testing qubits.
Our contribution is to provide a robust self-test for any odd prime dimension while self-testing qubits robustly is well-studied and understood~\cite{PhysRevA.91.052111,PhysRevA.93.062121,McKague_2012,Wang_2016}
The pivotal element is the novel generalization of the CHSH experiment~\cite{meyer2024bell} that has a simple analytic form and directly uses measurements from operators of the Heisenberg-Weyl algebra.
We conjecture that self-testing most of the algebra separates this Bell test from previous works~\cite{Coladangelo2017,Cui2020generalizationof,Sarkar2021} in higher dimensions that could not prove robustness.
This rises the question is there exists a minimal approach for robustness.

A further strength of our result is that prime-dimensional certification lifts to arbitrary dimension.
Since any finite-dimensional Hilbert space decomposes uniquely into a tensor product of prime-power subsystems, a robust self-test for each prime dimension yields a robust self-test for every composite dimension $d$.
Concretely, if $d=\prod_j p_j^{k_j}$ with primes $p_j$, the system can be regarded as $\bigotimes_j (\mathbb{C}_{p_j})^{\otimes k_j}$.
Our result guarantees that, up to a local isometries, each prime block and its associated measurement set are certified whenever the observed correlations are within $\epsilon$ of the ideal ones.
Because robustness bounds are stable under tensor products—errors accumulating at most linearly in the number of blocks—the full system inherits an $\mathcal{O}(\sum_{j} k_{j} \epsilon)$ robustness.
In other words, by self-testing every prime component independently and then invoking standard stability results for product states and measurements, one certifies the joint state and measurement operators for any composite dimension.
Hence self-testing for $d$ prime is not just sufficient but universal: it provides the modular building blocks from which robust self-tests for all higher-dimensional qudit systems can be assembled.

We want to highlight that in any Bell experiment, there exists a dichotomy of state and measurement captured by the local isometry.
In this way, our protocol also captures self-testing the (not-rotated) maximally entangled Bell state while identifying the Heisenberg-Weyl algebra up to a joint local rotation.
Related work~\cite{PhysRevLett.119.040402} established that only the maximally entangled state can maximally violate certain Bell inequalities, but without extending this to a full, robust self-test valid in arbitrary dimension.

Moreover, the results also verify that the correlations from the Bell operator~\eqref{eq:belloperator5self} are indeed nonlocal, which prior to this work could only be certified numerically for small dimensions.
Specifically, the self-test of the rotated Bell state using the Bell operator~\eqref{eq:belloperator5self} leads to commutation relations that are non-classical.
The positivity of the SOPO decomposition proves that strategies for lower or equal dimension cannot exceed the expectation value.

As an open question, one might consider unitary operators as the non-classical resource that are not diagonal in any of the Heisenberg-Weyl basis and complete the algebra, which could improve the robustness of the self-test.

Lastly, note that our methods can be extended to parallel self-testing~\cite{PhysRevA.93.062121,McKague_2016}.
In~\cite{meyer2024bell}, we consider nonlocality in qudit systems assuming that at least two copies of the state and local joint measurements are available.
This setup allows for particularly strong nonlocal violations based on local contextuality and Wigner negativity, which possibly improves the experiment's noise-robustness.

\paragraph*{Acknowledgements.} We acknowledge funding from Horizon Europe Research and Innovation Actions under Grant Agreement 101080173 CLUSTEC and the Quantum Internet Alliance (QIA) under Grant Agreements 820445, 101080128 and 101102140, as well as from the Plan France 2030 through the ANR-22-PETQ-0006-NISQ2LSQ and ANR-22-PETQ-0007 EPiQ projects.

\bibliographystyle{unsrturl}
\bibliography{main}

\pagebreak
\appendix

\onecolumngrid

\section{Characteristic function}
\label{sec:character5}

The characteristic function used in the Bell operator\,\eqref{eq:belloperator5self} is
\begin{align}
\chi[\left(U\otimes \mathds{1}\right)\vert \Phi \rangle]_{(x_{1},z_{1}),(x_{2},z_{2})} &= \omega^{2^{-1}(x_{1}z_{1}+x_{2}z_{2})} \langle \psi \vert U X^{-x_{1}} U^{\dagger} Z^{-z_{1}}\otimes X^{-x_{2}} Z^{-z_{2}} \vert \psi \rangle \\
&= \frac{1}{d} \, \omega^{2^{-1}(x_{1}z_{1}+z_{2}x_{2})} \sum_{j,k \in \fld}\langle jj \vert U^{\dagger} X^{-x_{1}} U Z^{-z_{1}}\otimes X^{-x_{2}} Z^{-z_{2}} \vert kk \rangle \\
&= \frac{1}{d} \, \omega^{2^{-1}(x_{1}z_{1}+x_{2}z_{2})} \sum_{j,k \in \fld} \omega^{ -(z_{1} + z_{2}) k}\langle jj \vert U^{\dagger} X^{-x_{1}} U \otimes X^{-x_{2}} \vert kk \rangle \\
&= \frac{1}{d} \, \omega^{2^{-1}(x_{1}z_{1}+x_{2}z_{2})} \sum_{j,k \in \fld} \omega^{ -(z_{1} + z_{2}) k} \omega^{ \nu_{k} - \nu_{k-x_{1}}} \langle jj \vert X^{-x_{1}} \otimes X^{-x_{2}} \vert kk \rangle \\
&= \frac{1}{d} \, \omega^{2^{-1}(x_{1}z_{1}+x_{2}z_{2})} \sum_{j,k \in \fld} \omega^{ -(z_{1} + z_{2}) k} \omega^{ \nu_{k} - \nu_{k-x_{1}}} \langle jj \vert k-x_{1}, k-x_{2} \rangle \\
&= \frac{1}{d}\, \delta_{x := x_{1} = x_{2}} \sum_{k \in \fld} \omega^{ (z_{1} + z_{2}) (2^{-1}x-k) + \nu_{k} - \nu_{k-x} } \\
&= \frac{1}{d}\, \delta_{x := x_{1} = x_{2}} \sum_{k \in \fld} \omega^{ (z_{1} + z_{2})k + \nu_{k+ 2^{-1}x} -\nu_{k- 2^{-1}x}}\,.
\end{align}
In particular, given $\nu_{k} = 12^{-1} (k - 3k^{2} + 2k^3)$ as in~\cite{Howard2012} for $d>3$,
\begin{equation}
\chi[\left(U\otimes \mathds{1}\right)\vert \Phi \rangle]_{(x_{1},z_{1}),(x_{2},z_{2})} = \delta_{x := x_{1} = x_{2}} \frac{\epsilon_{d}}{\sqrt{d}} \Big(\tfrac{-2^{-1}x}{d}\Big)\, \omega^{(2x)^{-1}(z_{1}+z_{2})^{2} -2^{-1}(z_{1}+z_{2}) - 24^{-1}(x^3 - x)}\,,
\end{equation}
with the Legendre symbol $\big(\tfrac{\cdot}{d}\big)$ and $\epsilon_d = \begin{cases} 1\,, \qquad d \equiv 1 \bmod 4 \\ i\,,\, \qquad d \equiv 3 \bmod 4 \end{cases}$.

\section{Sum-of-positive-operators decomposition}
\label{sec:decomp}
\paragraph*{For $d>3$:}
We show that there are two sum-of-positive-operators (SOPO) decomposition for the Bell operator\,\eqref{eq:belloperator5self},
\begin{equation}
d(d-1) \mathds{1} - \mathcal{B}_{d} = \sum_{n = 1}^{\tfrac{d-1}{2}} \sum_{j \in \fld} C_{n,j}^{\dagger} C_{n,j}^{\phantom{\dagger}} = \sum_{n = 1}^{\tfrac{d-1}{2}} \sum_{j \in \fld} D_{n,j}^{\dagger} D_{n,j}^{\phantom{\dagger}} \,,
\end{equation}
with
\begin{equation}
C_{n,j} := A_{j}^{n} - \sum_{ k \in \fld } g(j,k,n)^{\ast} \left( B_{k}^{n} \right)^{\dagger} \,, \qquad D_{n,j} := B_{j}^{n} - \sum_{ k \in \fld } g(j,k,n)^{\ast} \left( A_{k}^{n} \right)^{\dagger} \label{eq:sos-j} \,,
\end{equation}
and
\begin{equation}
g(j,k,n) = \frac{1}{d}\sum_{s \in \fld} \omega^{ n(j+k)s + \nu_{s+ 2^{-1}n} - \nu_{s- 2^{-1}n}} \,.
\end{equation}
For brevity, we write $C_{n,j} = A_{j}^{n} - P_{n,j}^{\dagger}$ with $ P_{n,j} = \sum_{ k \in \fld } g(j,k,n) B_{k}^{n}$.
Then, we regard the terms
\begin{align}
\label{ali:generalsosD0gen} P_{n,j} P_{-n,j} &= \sum_{ k,l \in \fld } g(j,k,n) g(j,l,n)^{\ast} B_{k}^{n} \left(B_{l}^{n} \right)^{\dagger} \,,
\end{align}
and calculate, for $n\neq 0$, that
\begin{align}
\sum_{j \in \fld} g(j,k,n) g(j,l,n)^{\ast} &= \frac{1}{d^{2}} \sum_{j,s,t \in \fld } \omega^{ n(j+k)s + \nu_{s+ 2^{-1}n} -\nu_{s- 2^{-1}n} - n(j+l)t - \nu_{t+ 2^{-1}n} + \nu_{t- 2^{-1}n}} \\
&= \frac{1}{d} \sum_{s,t \in \fld } \delta_{s,t}\, \omega^{ nks + \nu_{s+ 2^{-1}n} -\nu_{s- 2^{-1}n} - nlt - \nu_{t+ 2^{-1}n} + \nu_{t- 2^{-1}n}} \\
&= \frac{1}{d} \sum_{s \in \fld } \omega^{ ns (k-l) } \\
\label{eq:verifydecomp-gg}&= \delta_{k,l} \,.
\end{align}
With this, the decomposition is
\begin{align}
\sum_{n = 1}^{\tfrac{d-1}{2}}\sum_{j \in \fld} C_{n,j}^{\dagger} C_{n,j}^{\phantom{\dagger}} &= \sum_{n = 1}^{\tfrac{d-1}{2}} \sum_{j \in \fld} \left( \left( A_{j}^{n} \right)^{\dagger} - P_{n,j} \right) \left(A_{n,j} - \left( P_{n,j} \right)^{\dagger} \right) \\
\label{eq:verifydecom-1}&= \sum_{n = 1}^{\tfrac{d-1}{2}} \sum_{j \in \fld} \left( \left( A_{j}^{n} \right)^{\dagger} A_{j}^{n} - \left( A_{j}^{n} \right)^{\dagger} \left( P_{-n,j} \right)^{\dagger} - A_{j}^{n} P_{n,j} + P_{n,j} \left( P_{-n,j} \right)^{\dagger} \right) \\
\label{eq:verifydecom-2}&= \frac{d(d-1)}{2} - \sum_{n \in \fld^{\ast}} \sum_{ j,k \in \fld } g(j,k,n) A_{j}^{n} B_{k}^{n} + \sum_{n = 1}^{\tfrac{d-1}{2}} \sum_{ j,k,l \in \fld } g(j,k,n) g(j,l,n)^{\ast} B_{k}^{n} \left( B_{l}^{n} \right)^{\dagger} \\
\label{eq:verifydecom-3}&= \frac{d(d-1)}{2} - \mathcal{B}_{d} + \sum_{n = 1}^{\tfrac{d-1}{2}} \sum_{ k,l \in \fld } \delta_{k,l} \,B_{k}^{n} \left( B_{l}^{n} \right)^{\dagger} \\
& = d(d-1) \mathds{1} - \mathcal{B}_{d} \,.
\end{align}
From Eq.\,\eqref{eq:verifydecom-1} to Eq.\,\eqref{eq:verifydecom-2}, we used $A_{j}^{n} \left( A_{j}^{n} \right)^{\dagger} = \mathds{1}$ for all $j,n \in \fld$, and Eq.\,\eqref{ali:generalsosD0gen}.
From Eq.\,\eqref{eq:verifydecom-2} to Eq.\,\eqref{eq:verifydecom-3}, we used that the cross terms $A_{j}^{n} P_{n,j}$ constitute the Bell operator\,\eqref{eq:belloperator5self}, as well as Eq.\,\eqref{eq:verifydecomp-gg}.
Lastly, $B_{k}^{n} \left( B_{l}^{n} \right)^{\dagger} = \mathds{1}$, for all $k,n \in \fld$, verifies the decomposition.

Similarly, it is
\begin{align}
\sum_{n = 1}^{\tfrac{d-1}{2}}\sum_{j \in \fld} D_{n,j}^{\dagger} D_{n,j}^{\phantom{\dagger}} &= \frac{d(d-1)}{2} - \sum_{n \in \fld^{\ast}} \sum_{ j,k \in \fld } g(j,k,n) A_{j}^{n} B_{k}^{n} + \sum_{n = 1}^{\tfrac{d-1}{2}} \sum_{ j,k,l \in \fld } g(j,k,n) g(j,l,n)^{\ast} A_{k}^{n} \left( A_{l}^{n} \right)^{\dagger} \\
&= \frac{d(d-1)}{2} - \mathcal{B}_{d} + \sum_{n = 1}^{\tfrac{d-1}{2}} \sum_{ k,l \in \fld } \delta_{k,l} \,A_{k}^{n} \left( A_{l}^{n} \right)^{\dagger} \\
& = d(d-1) \mathds{1} - \mathcal{B}_{d} \,.
\end{align}

\paragraph*{For $d=3$:}
We use the operators Eq.\,\eqref{eq:sos3-j}, and calculate 
\begin{align}
& C_{0}^\dagger C_{0}^{\phantom{\dagger}} + C_{1}^\dagger C_{1}^{\phantom{\dagger}} + C_{2}^\dagger C_{2}^{\phantom{\dagger}} \\
&= \left[ A_{0}^{2} - \frac{1}{\sqrt{3}}\left( e^{\ii \phi_{1}} B_{0} + e^{\ii \phi_{2}} B_{1} + e^{\ii \phi_{1}} B_{2} \right) \right]\left[ A_{0} - \frac{1}{\sqrt{3}}\left( e^{-\ii \phi_{1}} B_{0}^{2} + e^{-\ii \phi_{2}} B_{1}^{2} + e^{-\ii \phi_{1}} B_{2}^{2} \right) \right] \nonumber\\
&+ \left[ A_{1}^{2} - \frac{1}{\sqrt{3}}\left( e^{\ii \phi_{2}} B_{0} + e^{\ii \phi_{1}} B_{1} + e^{\ii \phi_{1}} B_{2} \right) \right]\left[ A_{1} - \frac{1}{\sqrt{3}}\left( e^{-\ii \phi_{2}} B_{0}^{2} + e^{-\ii \phi_{1}} B_{1}^{2} + e^{-\ii \phi_{1}} B_{2}^{2} \right) \right] \nonumber \\
&+ \left[ A_{2}^{2} - \frac{1}{\sqrt{3}}\left( e^{\ii \phi_{1}} B_{0} + e^{\ii \phi_{1}} B_{1} + e^{\ii \phi_{2}} B_{2} \right) \right]\left[ A_{2} - \frac{1}{\sqrt{3}}\left( e^{-\ii \phi_{1}} B_{0}^{2} + e^{-\ii \phi_{1}} B_{1}^{2} + e^{-\ii \phi_{2}} B_{2}^{2} \right) \right]\\
&= 6 - \mathcal{B}_{3} \nonumber\\
&\phantom{=}~+\frac{1}{3} \Big( e^{-\ii (\phi_{1} - \phi_{2} )} B_{1} B_{0}^{2} + B_{2} B_{0}^{2} + e^{\ii (\phi_{1} - \phi_{2} )} B_{1}^{2} B_{0} + e^{\ii (\phi_{1} - \phi_{2} )} B_{2} B_{1}^{2} + B_{0} B_{2}^{2} e^{-\ii (\phi_{1} - \phi_{2} )} + B_{1} B_{2}^{2} \nonumber \\
&\phantom{\frac{1}{3}\Big(}+ e^{\ii (\phi_{1} - \phi_{2} )} B_{1} B_{0}^{2} + e^{\ii (\phi_{1} - \phi_{2} )} B_{2} B_{0}^{2} + e^{-\ii (\phi_{1} - \phi_{2} )} B_{1}^{2} B_{0} + B_{2} B_{1}^{2} + e^{-\ii (\phi_{1} - \phi_{2} )} B_{0} B_{2}^{2} + B_{1} B_{2}^{2} \nonumber\\
&\phantom{\frac{1}{3}\Big(}+ B_{1} B_{0}^{2} + e^{-\ii (\phi_{1} - \phi_{2} )} B_{2} B_{0}^{2} + B_{1}^{2} B_{0} + e^{-\ii (\phi_{1} - \phi_{2} )} B_{2} B_{1}^{2} + e^{\ii (\phi_{1} - \phi_{2} )} B_{0} B_{2}^{2} + e^{\ii (\phi_{1} - \phi_{2} )} B_{1} B_{2}^{2} \Big) \\
&= 6 - \mathcal{B}_{3} \nonumber\\ 
&\phantom{=}~+\frac{1}{3} \Big( \omega^{2} B_{1} B_{0}^{2} + B_{2} B_{0}^{2} + \omega B_{1}^{2} B_{0} + \omega B_{2} B_{1}^{2} + B_{0} B_{2}^{2} \omega^{2} + B_{1} B_{2}^{2} \nonumber \\
&\phantom{\frac{1}{3}\Big(}+ \omega B_{1} B_{0}^{2} + \omega B_{2} B_{0}^{2} + \omega^{2} B_{1}^{2} B_{0} + B_{2} B_{1}^{2} + \omega^{2} B_{0} B_{2}^{2} + B_{1} B_{2}^{2} \nonumber\\
&\phantom{\frac{1}{3}\Big(}+ B_{1} B_{0}^{2} + \omega^{2} B_{2} B_{0}^{2} + B_{1}^{2} B_{0} + \omega^{2} B_{2} B_{1}^{2} + \omega B_{0} B_{2}^{2} + \omega B_{1} B_{2}^{2} \Big) \\
&= 6 - \mathcal{B}_{3}
\end{align}
In the last line, we used that $1+ \omega + \omega^{2} = 0$.
Equivalently, we can also show that
\begin{align}
& D_{0}^\dagger D_{0}^{\phantom{\dagger}} + D_{1}^\dagger D_{1}^{\phantom{\dagger}} + D_{2}^\dagger D_{2}^{\phantom{\dagger}}\\
&= \left[ B_{0}^{2} - \frac{1}{\sqrt{3}}\left( e^{\ii \phi_{1}} A_{0} + e^{\ii \phi_{2}} A_{1} + e^{\ii \phi_{1}} A_{2} \right) \right]\left[ B_{0} - \frac{1}{\sqrt{3}}\left( e^{-\ii \phi_{1}} A_{0}^{2} + e^{-\ii \phi_{2}} A_{1}^{2} + e^{-\ii \phi_{1}} A_{2}^{2} \right) \right]\nonumber \\
&+ \left[ B_{1}^{2} - \frac{1}{\sqrt{3}}\left( e^{\ii \phi_{2}} A_{0} + e^{\ii \phi_{1}} A_{1} + e^{\ii \phi_{1}} A_{2} \right) \right]\left[ B_{1} - \frac{1}{\sqrt{3}}\left( e^{-\ii \phi_{2}} A_{0}^{2} + e^{-\ii \phi_{1}} A_{1}^{2} + e^{-\ii \phi_{1}} A_{2}^{2} \right) \right]\nonumber \\
&+ \left[ B_{2}^{2} - \frac{1}{\sqrt{3}}\left( e^{\ii \phi_{1}} A_{0} + e^{\ii \phi_{1}} A_{1} + e^{\ii \phi_{2}} A_{2} \right) \right] \left[ B_{2} - \frac{1}{\sqrt{3}}\left( e^{-\ii \phi_{1}} A_{0}^{2} + e^{-\ii \phi_{1}} A_{1}^{2} + e^{-\ii \phi_{2}} A_{2}^{2} \right) \right]\\
&= 6 - \mathcal{B}_{3} \,.
\end{align}

\section{Twisted commutation relations}
\label{sec:comrel}

\paragraph*{For $d>3$:}

For $g(j,k,n)$ from Eq.\,\eqref{eq:gdefinition}, we show the \textit{quadratic folding identity}:
\begin{equation}
g\left(j,k,n+n'\right) =\sum_{r\in \fld} g\left(j,k + n'r, n\right) g\left(j,k - n r, n'\right) \omega^{2^{-1}nn'(n+n')r}
\label{eq:qfid}
\end{equation}
by calculating
\begin{align}
&\sum_{r\in \fld} g\left(j,k + n'r, n\right) g\left(j,k - n r, n'\right) \omega^{2^{-1}nn'(n+n')r} \\
&= \frac{1}{d^{2}} \sum_{r,s,t\in \fld} \omega^{n(j+k+n'r)s +\nu_{s+2^{-1}n} - \nu_{s-2^{-1}n} } \omega^{n'(j+ k - n r)t +\nu_{t+2^{-1}n'} - \nu_{t-2^{-1}n'} } \omega^{2^{-1}nn'(n+n')r} \\
&= \frac{1}{d} \sum_{s,t\in \fld} \delta_{nn'(s-t+2^{-1}(n+n'))} \omega^{n(j+k)s +\nu_{s+2^{-1}n} - \nu_{s-2^{-1}n} } \omega^{n'(j+ k)t +\nu_{t+2^{-1}n'} - \nu_{t-2^{-1}n'} } \\
&= \frac{1}{d} \sum_{s\in \fld} \omega^{n(j+k)s +\nu_{s+2^{-1}n} - \nu_{s-2^{-1}n} + n'(j+ k)(s+2^{-1}(n+n')) +\nu_{s+2^{-1}(n+n')+2^{-1}n'} - \nu_{s+2^{-1}(n+n')-2^{-1}n'} } \\
&= \frac{1}{d} \sum_{s\in \fld} \omega^{(n+n')(j+k)(s+2^{-1} n') -\nu_{s - 2^{-1}n} - \nu_{s+2^{-1}(n+2n')} } \\
&= \frac{1}{d} \sum_{ s \in \fld} \omega^{(n+n')(j+k)s +\nu_{s+2^{-1}(n+n')} - \nu_{s-2^{-1}(n+n')} } \\
&= g\left(j,k, n+n'\right) \,.
\end{align}
Then, since $\left( A_{j}^{n\phantom{'}} \right)^{\dagger} \left( A_{j}^{n'} \right)^{\dagger} = \left( A_{j}^{n+n'} \right)^{\dagger}$, Eq.\,\eqref{eq:abrelation} and Eq.\,\eqref{eq:qfid} lead to
\begin{align}
\label{eq:calcBcomgen-0}&\left[ \left( A_{j}^{n\phantom{'}} \right)^{\dagger} \left( A_{j}^{n'} \right)^{\dagger} - \left( A_{j}^{n+n'} \right)^{\dagger} \right] \vert \psi \rangle \\[0.1cm]
\label{eq:calcBcomgen-1}&= \sum_{ k,l \in \fld } g(j,k,n) g(j,l,n') B_{k}^{n} B_{l}^{n'} \vert \psi \rangle - \sum_{ m \in \fld} g(j,m,n+n') B_{m}^{n+n'} \vert \psi \rangle\\
\label{eq:calcBcomgen-2}&= \sum_{ k,l \in \fld } g(j,k,n) g(j,l,n') B_{k}^{n} B_{l}^{n'} \vert \psi \rangle - \sum_{ m,r \in \fld} g\left(j,m + n'r, n\right) g\left(j,m - n r, n'\right) \omega^{2^{-1}nn'(n+n')r} B_{m}^{n+n'} \vert \psi \rangle \\
\label{eq:calcBcomgen-3} &= \sum_{ k,l \in \fld } g(j,k,n) g(j,l,n') \left( B_{k}^{n} B_{l}^{n'} - \omega^{2^{-1}nn'(k-l)} B_{(n+n')^{-1}(nk+n'l)}^{n+n'} \right) \vert \psi \rangle \\
\label{eq:calcBcomgen-4}&\overset{!}{=} 0 \,,
\end{align}
using the substitution $m+n'r = k,m-nr=l$ such that $(n'+n)r = k-l,(n+n')m=(nk+n'l)$ from Eq.\,\eqref{eq:calcBcomgen-2} to Eq.\,\eqref{eq:calcBcomgen-3}.
Since $\nu$ is not a polynomial of degree less than $3$, $g(j,k,n) \neq 0$ for any $j,k \in\fld$ and $n \in \fld^{\ast}$.
As a result, it is
\begin{equation}
\label{eq:comrelhalfgen}
\left( B_{k}^{n} B_{l}^{n'} - \omega^{2^{-1}nn'(k-l)} B_{(n+n')^{-1}(nk+n'l)}^{n+n'} \right) \vert \psi \rangle = 0
\end{equation}
for all $k,l \in \fld$ and $n,n',n+n' \in \fld^{\ast}$.
In fact one can derive the expression for general $n,n'$ from the case $n=n'=1$.
It will be beneficial however, to use the general case when $\epsilon > 0$ in Lemma~\ref{lem:com-ideal}.
Furthermore, the above can be shown for the operators $A_{j}$, with the exact same reasoning by swapping the operators $B_{j}$ and $A_{j}$, such that
\begin{equation}
\label{eq:comrelhalf}
\left(A_{k} A_{l} - \omega^{2^{-1}(k-l)} A_{2^{-1}(k+l)}^{2} \right) \vert \psi \rangle = 0 \,, \qquad \left(B_{k} B_{l} - \omega^{2^{-1}(k-l)} B_{2^{-1}(k+l)}^{2} \right) \vert \psi \rangle = 0\,,
\end{equation}
for all $l,k \in \fld$.

\paragraph*{For $d=3$:}
We complement the proof by the missing steps, i.e.~how Eq.\,\eqref{eq:heisenbergweylalgebra} follows from Eq.\,\eqref{eq:bconditions}.
From $\langle \psi \vert 6 \mathds{1} - \mathcal{B}_{3} \vert \psi \rangle = 0$, it follows $C_{j} \vert \psi \rangle = 0$ in Eq.\,\eqref{eq:sos3-j} or 
\begin{align}
A_{0}\vert \psi \rangle &= \frac{e^{-\ii \phi_{1}}}{\sqrt{3}} \left( B_{0}^{\dagger} + \omega B_{1}^{\dagger} + B_{2}^{\dagger} \right) \vert \psi \rangle \,, \\
A_{1}\vert \psi \rangle &= \frac{e^{-\ii \phi_{1}}}{\sqrt{3}} \left( \omega B_{0}^{\dagger} + B_{1}^{\dagger} + B_{2}^{\dagger} \right) \vert \psi \rangle \,, \\
A_{2}\vert \psi \rangle &= \frac{e^{-\ii \phi_{1}}}{\sqrt{3}} \left( B_{0}^{\dagger} + B_{1}^{\dagger} + \omega B_{2}^{\dagger} \right) \vert \psi \rangle \,,
\end{align}
where we used the first constraint from Eq.\,\eqref{eq:phaseDifference}, i.e.~$\phi_{1}-\phi_{2} \equiv 2\pi/3$.
Since it must be $A_{j}^{2} = A_{j}^{\dagger}$ for $j\in \{0,1,2\}$, the above lead to the equations
\begin{align}
\left[ \frac{e^{-2\ii \phi_{1}}}{3} \left( B_{0}^{\dagger} + \omega B_{1}^{\dagger} + B_{2}^{\dagger} \right)^{2} - \frac{e^{\ii \phi_{1}}}{\sqrt{3}} \left( B_{0} + \omega^{2} B_{1} + B_{2} \right) \right] \vert \psi \rangle &= 0 \,, \\
\left[ \frac{e^{-2\ii \phi_{1}}}{3} \left( \omega B_{0}^{\dagger} + B_{1}^{\dagger} + B_{2}^{\dagger} \right)^{2} - \frac{e^{\ii \phi_{1}}}{\sqrt{3}} \left( \omega^{2} B_{0} + B_{1} + B_{2} \right) \right] \vert \psi \rangle &= 0 \,, \\
\left[ \frac{e^{-2\ii \phi_{1}}}{3} \left( B_{0}^{\dagger} + B_{1}^{\dagger} + \omega B_{2}^{\dagger} \right)^{2} - \frac{e^{\ii \phi_{1}}}{\sqrt{3}} \left( B_{0} + B_{1} + \omega^{2} B_{2} \right) \right] \vert \psi \rangle &= 0 \,.
\end{align}
The second constraint from Eq.\,\eqref{eq:phaseDifference}, i.e.~$1 - \sqrt{3} e^{3 \ii \phi_{1}} = \omega$ leads to the system of linear equations
\begin{align}
\left[ \omega \{ B_{0}^{2},B_{1}^{2}\} + \phantom{\omega} \{ B_{0}^{2},B_{2}^{2}\} + \omega \{ B_{1}^{2},B_{2}^{2}\} + \omega B_{0} + \phantom{\omega} B_{1} + \omega B_{2} \right] \vert \psi \rangle &= 0 \,, \\
\left[ \phantom{\omega} \{ B_{0}^{2},B_{1}^{2}\} + \omega \{ B_{0}^{2},B_{2}^{2}\} + \omega \{ B_{1}^{2},B_{2}^{2}\} + \phantom{\omega} B_{0} + \omega B_{1} + \omega B_{2} \right] \vert \psi \rangle &= 0 \,, \\
\left[ \omega \{ B_{0}^{2},B_{1}^{2}\} + \omega \{ B_{0}^{2},B_{2}^{2}\} + \phantom{\omega} \{ B_{1}^{2},B_{2}^{2}\} + \omega B_{0} + \omega B_{1} + \phantom{\omega} B_{2} \right] \vert \psi \rangle &= 0 \,.
\end{align}
Rearranging these leads to Eq.\,\eqref{eq:bconditions}, namely,
\begin{equation}
\left[ \{B_{0}^{2},B_{1}^{2}\} + B_{2} \right] \vert \psi \rangle = 0 \,,\quad
\left[ \{B_{0}^{2},B_{2}^{2}\} + B_{1} \right] \vert \psi \rangle = 0 \,,\quad
\left[ \{B_{1}^{2},B_{2}^{2}\} + B_{0} \right] \vert \psi \rangle = 0 \,.
\end{equation}
In the following, we omit the state $\vert \psi \rangle$ but note that all identities hold with respect to their action on it.
First, one finds that
\begin{equation}
B_{0} B_{1} B_{2} = - B_{0} B_{0}^{2} - B_{0} B_{2} B_{1} = - \mathds{1} - B_{0} B_{2} B_{1} = B_{2} B_{0} B_{1} = B_{1} B_{2} B_{0} \,, \label{eq:3bcyclic}
\end{equation}
and for the anti-cyclic permutations,
\begin{equation}
B_{0} B_{2} B_{1} = B_{2} B_{1} B_{0} = B_{1} B_{0} B_{2} \,.
\end{equation}
It immediately follows that
\begin{equation}
[B_{j},B_{0} B_{1} B_{2}] = 0\,,
\end{equation}
for all $j \in \{0,1,2\}$.
Recall the conjugation elements
\begin{equation}
Q:=B_{0}\,B_{1}\,B_{0}^{\dagger}B_{1}^{\dagger}\,, \quad Q^{\prime} = B_{2} B_{0} B_{2}^{\dagger} B_{0}^{\dagger}\,, \quad Q^{\prime\prime} = B_{1} B_{2} B_{1}^{\dagger} B_{2}^{\dagger} \,.
\end{equation}
Multiplying $B_{0}^{\dagger} B_{1}^{\dagger} + B_{1}^{\dagger} B_{0}^{\dagger} = -B_{2}$ from the right by $B_{0} B_{1} $ gives $\mathds{1} + Q = -B_{0} B_{1} B_{2}$.
The same hold of $Q^{\prime}$ and $Q^{\prime\prime}$, i.e.~$\mathds{1} + Q^{\prime} = -B_{2} B_{0} B_{1}$ and $\mathds{1} + Q^{\prime\prime} = -B_{1} B_{2} B_{0}$.
From Eq.\,\eqref{eq:3bcyclic}, we find $Q = Q^{\prime} = Q^{\prime\prime}$ as well as \[[B_{i},Q]=0 \,, \quad \text{for }i = 1,2,3 \,.\]
With the twisted commutation relations it is straightforward to calculate $Q^3 = \mathds{1}$.
As a result, $Q$ commutes with the whole algebra and must be a scalar multiple of the identity:
\begin{equation}
Q=\lambda\mathds{1}\,,\qquad\lambda^{3}=1\,,\qquad \lambda\neq 1\,.
\end{equation}
Choosing the orientation \(0\rightarrow 1 \rightarrow 2 \rightarrow 0 \) sets \(\lambda=\omega\) (the opposite orientation would give \(\omega^{2}\)).
With \(Q = \omega \mathds{1}\), the defining equation becomes the $\omega$-commutation rules
\begin{equation}
B_{0}B_{1}=\omega B_{1}B_{0},\quad
B_{1}B_{2}=\omega B_{2}B_{1},\quad
B_{2}B_{0}=\omega B_{0}B_{2}\,.
\end{equation}
from Eq.\,\eqref{eq:heisenbergweylalgebra}.

\section{Isomorphism}
\label{sec:iso}

We construct the local isometry $V_A \otimes V_B$ where each $V_A,V_B$ is described by the circuit in Fig.~\ref{fig:swapcircuit-tech}.
To prove Theorem~\ref{th:main-ideal}, we spell out the action of circuit layer-by-layer.
We start with two auxiliary qubits and the state $\vert \psi \rangle$, and apply Fourier transform on both auxiliary [1].
Then, apply $A_{0}^{\dagger} A_{1}^{\phantom{\dagger}}$ and $B_{0}^{\dagger} B_{1}^{\phantom{\dagger}}$ operators [2], and again the Fourier transform on both sides [3].
Lastly, apply $A_{0}^{-1}$ and $B_{0}^{-1}$ operators [4].
Namely, it is
\begin{align}
&\phantom{\Rightarrow} \vert \psi \rangle \otimes \vert 0,0 \rangle \nonumber \\[0.1cm]
&\overset{[1]}{\Rightarrow} \frac{1}{d} \sum_{j,k \in \fld} \vert j,k\rangle \otimes \vert \psi \rangle \label{eq:isocirc1} \\
&\overset{[2]}{\Rightarrow} \frac{1}{d} \sum_{j,k \in \fld} \vert j,k\rangle \otimes (\omega^{-2^{-1}j} A_{0}^{\dagger} A_{j}) (\omega^{2^{-1}k} B_{0}^{\dagger} B_{k}) \vert \psi \rangle \label{eq:isocirc2} \\
&\overset{[3]}{\Rightarrow} \frac{1}{d^{2}} \sum_{n,m \in \fld} \vert n,m\rangle \otimes \Big( \sum_{j \in \fld} \omega^{-jn} (\omega^{-2^{-1}j}A_{0}^{\dagger} A_{j}) \Big) \Big( \sum_{k \in \fld} \omega^{-km} (\omega^{-2^{-1}k} B_{0}^{\dagger} B_{k}) \Big) \vert \psi \rangle \label{eq:isocirc3} \\
&\overset{[4]}{\Rightarrow} \frac{1}{d^{2}} \sum_{n,m \in \fld} \vert n,m\rangle \otimes \Big( \left( A_{0}^{n} \right)^{\dagger} \sum_{j \in \fld} \omega^{-jn} (\omega^{-2^{-1}j} A_{0}^{\dagger} A_{j}) \Big) \Big( \left( B_{0}^{m} \right)^{\dagger} \sum_{k \in \fld} \omega^{-km} (\omega^{-2^{-1}k} B_{0}^{\dagger} B_{k}) \Big) \vert \psi \rangle \label{eq:isocirc4}
\end{align}
We evaluate Eq.\,\eqref{eq:isocirc4}, first by permuting the $A_{0}$ and $B_{0}$ with the $A_{1}$ and $B_{1}$ operators, then apply the relations from Eq.\,\eqref{eq:abrelation}.
For convenience, we denote $\Bar{B}_{m} := \left( B_{0}^{m} \right)^{\dagger} \sum_{k \in \fld} \omega^{-km} (\omega^{-2^{-1}k} B_{0}^{\dagger} B_{k}) $,
\begin{align}
\label{eq:isoeval-1} \left( V_{A} \otimes V_{B} \right) ( \vert \psi \rangle) &= \frac{1}{d^{2}} \sum_{n,m \in \fld} \vert n,m \rangle \otimes \Big( \left( A_{0}^{n} \right)^{\dagger} \sum_{j \in \fld} \omega^{-jn} (\omega^{-2^{-1}j} A_{0}^{\dagger} A_{j}) \Big) \Bar{B}_{m} \vert \psi \rangle \\
\label{eq:isoeval-2}&= \frac{1}{d^{2}} \sum_{n,m,j \in \fld} \vert n,m \rangle \otimes \omega^{-2^{-1}jn} A_{-n^{-1}j}^{-n} \Bar{B}_{m} \vert \psi \rangle \\
\label{eq:isoeval-3}&= \frac{1}{d^{3}} \sum_{n,m,j,l,s \in \fld} \vert n,m\rangle \otimes \Bar{B}_{m} \, \omega^{-2^{-1}jn} \omega^{ (j-l) s + \nu_{s+ 2^{-1}n} - \nu_{s- 2^{-1}n} } B_{n^{-1}l}^{n} \vert \psi \rangle \\
\label{eq:isoeval-4}&= \frac{1}{d^{2}} \sum_{n,m,l,s \in \fld} \vert n,m\rangle \otimes \Bar{B}_{m} \, \delta_{s, 2^{-1}n} \, \omega^{ -l s + \nu_{s+ 2^{-1}n} - \nu_{s- 2^{-1}n} } B_{n^{-1}l}^{n} \vert \psi \rangle \\
\label{eq:isoeval-5}&= \frac{1}{d^{2}} \sum_{n,m,l \in \fld} \omega^{ \nu_{n} - \nu_{0} } \vert n,m\rangle \otimes \Bar{B}_{m} \, \omega^{ -2^{-1} n l } B_{n^{-1}l}^{n} \vert \psi \rangle \\
\label{eq:isoeval-6}\Big(\,^{r=n-m}_{m=n+r} \Big)&= \frac{1}{d^{2}} \sum_{n,r,k,l \in \fld} \omega^{ \nu_{n} - \nu_{0} } \vert n, n+r \rangle \otimes \omega^{-kr} B_{0}^{-r} \omega^{-2^{-1}(k+l)n} B_{-n^{-1}k}^{-n} B_{n^{-1}l}^{n} \vert \psi \rangle \\
\label{eq:isoeval-7} &= \frac{1}{d^{2}} \sum_{n,r,k,l \in \fld} \omega^{ \nu_{n} - \nu_{0} } \vert n,n+r \rangle \otimes \omega^{-kr} B_{0}^{-r} \omega^{- 2^{-1}(l+k)} B_{0}^{\dagger} B_{k+l} \vert \psi \rangle \\
\label{eq:isoeval-8} \Big(\,^{t=k+l}_{l=t-k} \Big) &= \frac{1}{d^{2}} \sum_{n,r,k,t \in \fld} \omega^{ \nu_{n} - \nu_{0} } \vert n,n+r \rangle \otimes \omega^{-kr} B_{0}^{-r} \omega^{- 2^{-1}t} B_{0}^{\dagger} B_{t} \vert \psi \rangle \\
\label{eq:isoeval-9} &= \frac{1}{d} \sum_{n,r,t \in \fld} \omega^{ \nu_{n} - \nu_{0} } \vert n,n+r \rangle \otimes \delta_{r,0} B_{0}^{-r} \omega^{- 2^{-1}t} B_{0}^{\dagger} B_{t} \vert \psi \rangle \\
\label{eq:isoeval-10} &= \frac{1}{d} \sum_{n,t \in \fld} \omega^{ \nu_{n} - \nu_{0} } \vert n,n \rangle \otimes \omega^{- 2^{-1}t} B_{0}^{\dagger} B_{t} \vert \psi \rangle \,.
\end{align}
We identify $\vert aux \rangle = \tfrac{1}{\sqrt{d}} \sum_{t\in \fld} \omega^{-2^{-1}t} B_{0}^{\dagger} B_{t} \vert \psi \rangle$ and $ \left( U_{\nu} \otimes \mathds{1} \right) \vert\Phi \rangle = \tfrac{1}{\sqrt{d}} \sum_{n\in \fld}\omega^{\nu_{n}}\vert n,n\rangle$ with the Bell state $\vert \Phi \rangle$ and $U_{\nu} = \sum_{ k \in \fld} \omega^{\nu_{k}} \vert k \rangle \langle k \vert$ and a global phase $\omega^{\nu_{0} }$.
From Eq.\,\eqref{eq:isoeval-1} to Eq.\,\eqref{eq:isoeval-2}, use $A_{-n,0} A_{-1,0} A_{1,j} = \omega^{2^{-1}(n+1)j} A_{-n,j}$.
From Eq.\,\eqref{eq:isoeval-2} to Eq.\,\eqref{eq:isoeval-3}, use Eq.~\eqref{eq:abrelation}.
From Eq.\,\eqref{eq:isoeval-3} to Eq.\,\eqref{eq:isoeval-4}, sum over $j$.
From Eq.\,\eqref{eq:isoeval-4} to Eq.\,\eqref{eq:isoeval-5}, sum over $s$.
From Eq.\,\eqref{eq:isoeval-5} to Eq.\,\eqref{eq:isoeval-6}, unpack $\Bar{B}_{m}$ and substitute $r = n-m$ for $m$.
From Eq.\,\eqref{eq:isoeval-6} to Eq.\,\eqref{eq:isoeval-7}, use $\omega^{-2^{-1}(k+l)n} B_{-n^{-1}k}^{-n} B_{n^{-1}l}^{n} = \omega^{- 2^{-1}(l+k)} B_{0}^{\dagger} B_{k+l}$.
From Eq.\,\eqref{eq:isoeval-7} to Eq.\,\eqref{eq:isoeval-8}, substitute $t = k+l$ for $l$.
From Eq.\,\eqref{eq:isoeval-8} to Eq.\,\eqref{eq:isoeval-9}, sum over $k$.
From Eq.\,\eqref{eq:isoeval-9} to Eq.\,\eqref{eq:isoeval-10}, sum over $m$.

Then, we show that $\left( V_{A} \otimes V_{B} \right) ( A_{u}^{v} \vert \psi \rangle) = \left( T_{(v,vu)} \otimes \mathds{1} \right) \vert \psi \rangle \otimes \vert aux \rangle$:
\begin{align}
\label{eq:isoevalA-1} \left( V_{A} \otimes V_{B} \right) ( A_{u}^{v} \vert \psi \rangle) &= \frac{1}{d^{2}} \sum_{n,m \in \fld} \vert n,m \rangle \otimes \Big( \left( A_{0}^{n} \right)^{\dagger} \sum_{j \in \fld} \omega^{-jn} (\omega^{-2^{-1}j} A_{0}^{\dagger} A_{j}) \Big) \Bar{B}_{m} \left( A_{u}^{v} \vert \psi \rangle \right)\\
\label{eq:isoevalA-2}&= \frac{1}{d^{2}} \sum_{n,m,j \in \fld} \vert n,m \rangle \otimes \omega^{-2^{-1}j(n-v) + 2^{-1}nu} A_{(-n+1)^{-1}(j+uv)}^{-n+v} \Bar{B}_{m} \vert \psi \rangle \\
\label{eq:isoevalA-3}\Big(\,^{o=n-v}_{n=o+v} \Big)&= \frac{1}{d^{2}} \sum_{o,m,j \in \fld} \vert o+v,m \rangle \otimes \omega^{-2^{-1}jo + 2^{-1}(o+v)uv} A_{-o^{-1}(j+uv)}^{-o} \Bar{B}_{m} \vert \psi \rangle \\
\label{eq:isoevalA-4}\Big(\,^{j'=j+uv}_{j=j'-uv} \Big) &= \frac{1}{d^{2}} \sum_{o,m,j' \in \fld} \vert o+v,m \rangle \otimes \omega^{-2^{-1}(j'-uv)o + 2^{-1}(o+v)uv} A_{-o^{-1}j'}^{-o} \Bar{B}_{m} \vert \psi \rangle \\
\label{eq:isoevalA-5} &= \frac{1}{d} \sum_{o,t \in \fld} \omega^{ \nu_{o} - \nu_{0} } \omega^{2^{-1}uv^2} X^{v} Z^{uv} \vert o,o \rangle \otimes \omega^{- 2^{-1}t} B_{0}^{\dagger} B_{t} \vert \psi \rangle \\
\label{eq:isoevalA-6} &= \frac{1}{d} \sum_{o,t \in \fld} \omega^{ \nu_{o} - \nu_{0} } \left( T_{(v,vu)} \otimes \mathds{1} \right) \vert o,o \rangle \otimes \omega^{- 2^{-1}t} B_{0}^{\dagger} B_{t} \vert \psi \rangle \\
\label{eq:isoevalA-7} &= \left( T_{(v,vu)} \otimes \mathds{1} \right) \vert \psi \rangle \otimes \vert aux \rangle \,.
\end{align}
From Eq.~\eqref{eq:isoevalA-1} to Eq.~\eqref{eq:isoevalA-2}, we use that $ \omega^{-jn-2^{-1}j} \left( A_{0}^{n} \right)^{\dagger} A_{0}^{\dagger} A_{j} =\omega^{-2^{-1}j(n-v) + 2^{-1}nu} A_{(-n+1)^{-1}(j+uv)}^{-n+v}$ from Eq.\,\eqref{eq:lemcomrel}.
From Eq.~\eqref{eq:isoevalA-2} to Eq.~\eqref{eq:isoevalA-3}, we substitute $o=n-v$ for $n$.
From Eq.~\eqref{eq:isoevalA-3} to Eq.~\eqref{eq:isoevalA-4}, we substitute $j'=j+uv$ for $j$.
From Eq.~\eqref{eq:isoevalA-4} to Eq.~\eqref{eq:isoevalA-5}, apply the steps from Eqs.\,\eqref{eq:isoeval-2}\,-\,\eqref{eq:isoeval-10}.
From Eq.~\eqref{eq:isoevalA-5} to Eq.~\eqref{eq:isoevalA-6}, we identify the operators $T_{(v,vu)}$.
From Eq.~\eqref{eq:isoevalA-6} to Eq.~\eqref{eq:isoevalA-7}, we identify the states $\vert \psi \rangle$ and $\vert aux \rangle$.
Due to the symmetry of the state and the operators $A_{j} ~ B_{j}$, we can show in the exact same fashion that $\left( V_{A} \otimes V_{B} \right) ( B_{u}^{v} \vert \psi \rangle) = \left( \mathds{1} \otimes T_{(v,vu)} \right) \vert \psi \rangle \otimes \vert aux \rangle$.

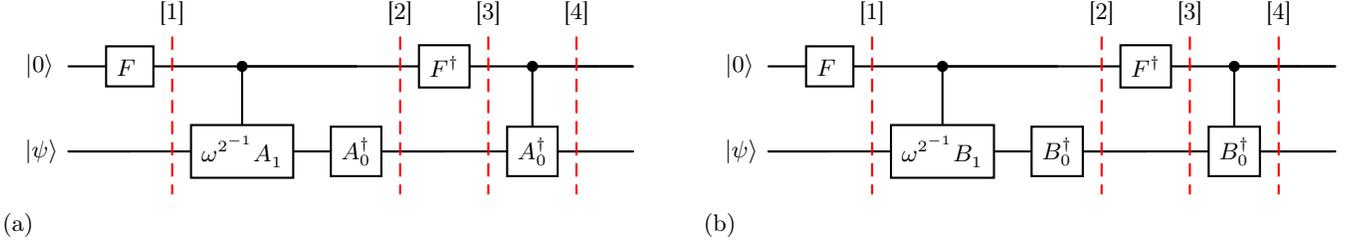
\begin{figure}[ht]
\centering
\begin{subfigure}[b]{0.48\textwidth}
\begin{quantikz}
\lstick{$\ket{ 0}\,$} & \gate{F\phantom{'}} \slice{[1]} & \ctrl{1} & \qw \slice{[2]} & \gate{F^{\dagger}} \slice{[3]} & \ctrl{1} \slice{[4]} & \qw & \qw \\
\lstick{$\ket{\psi}$}& \qw & \gate{\omega^{2^{-1}}A_{1}^{\phantom{\dagger}}} & \gate{A_{0}^{\dagger}} & \qw & \gate{A_{0}^{\dagger}} &\qw & \qw
\end{quantikz}
\caption{}
\label{fig:swapcircuit-tech-a}
\end{subfigure}
\hfill
\begin{subfigure}[b]{0.48\textwidth}
\begin{quantikz}
\lstick{$\ket{ 0}\,$} & \gate{F\phantom{'}} \slice{[1]} & \ctrl{1} & \qw \slice{[2]} & \gate{F^{\dagger}} \slice{[3]} & \ctrl{1} \slice{[4]} & \qw & \qw \\
\lstick{$\ket{\psi}$}& \qw & \gate{\omega^{2^{-1}} B_{1}^{\phantom{\dagger}}} & \gate{B_{0}^{\dagger}} & \qw & \gate{B_{0}^{\dagger}} &\qw & \qw
\end{quantikz}
\caption{}
\label{fig:swapcircuit-tech-b}
\end{subfigure}
\caption{Two circuits describing the action of the isometry $V_{A}$ in~\ref{fig:swapcircuit-tech-a} and $V_{B}$ in~\ref{fig:swapcircuit-tech-b}, both on the states $\vert \psi \rangle$ and $\vert 0 \rangle$, with the Fourier operation (, a generalized Hadamard matrix), $F = \tfrac{1}{\sqrt{d}} \sum_{i,j \in \fld} \omega^{jk} \vert j \rangle \langle k \vert$, and operators $A_{0},A_{1},B_{0},B_{1}$. The slices [1] - [4] correspond to Eqs.\,\eqref{eq:isocirc1} - \eqref{eq:isocirc4}.}
\label{fig:swapcircuit-tech}
\end{figure}

\section{Robust self-testing relations ($d > 3$)}
\label{sec:robust}

Under the assumptions of Lemma~\ref{lem:com-ideal}, $ \langle \psi \vert d(d-1) \mathds{1} - \mathcal{B}_{d} \vert \psi \rangle \leq \epsilon$ for an $\epsilon >0$.
The operator has a sum-of-positive-operators (SOPO) decomposition: $ d(d-1) \mathds{1} - \mathcal{B}_{d} = \sum_{n = 1}^{(d-1)/2} \sum_{j \in \fld} C_{n,j}^{\dagger} C_{n,j}^{\phantom{\dagger}}$ with $C_{n,j}$ from Eq.\,\eqref{eq:sos-j}.
If follows that $\langle \psi \vert C_{n,j}^{\dagger} C_{n,j}^{\phantom{\dagger}} \vert \psi \rangle \leq \epsilon$, and consequently,
\begin{equation}
\sqrt{ \langle \psi \vert C_{n,j}^{\dagger} C_{n,j}^{\phantom{\dagger}} \vert \psi \rangle } = \Vert C_{n,j} \vert \psi \rangle \Vert_{1} \leq \sqrt{\epsilon}
\label{eq:boundccdagger5}
\end{equation}
for all $j\in \fld$, $n \in \fld^{\ast}$, and the trace norm $\Vert \rho \Vert_{1} := \tr \sqrt{ \rho^{\dagger} \rho}$.
Form here onward, we follow the steps of the proof for the ideal case and bound the terms one after the other.

Recall that $C_{n,j} = A_{j}^{n} - P_{n,j}^{\dagger}$ and $P_{n,j} = \sum_{k\in\fld} g(j,k,n) B^{n}_{k}$ with $g(j,k,n) = \tfrac{1}{d}\sum_{s \in \fld}\omega^{ n(j+k)s + \nu_{s+ 2^{-1}n} - \nu_{s- 2^{-1}n}}$.
Applying $\left( A_{j}^{n\phantom{'}} \right)^{\dagger} \left( A_{j}^{n'} \right)^{\dagger} = \left( A_{j}^{n+n'} \right)^{\dagger}$,
\begin{equation}
\begin{aligned}
&\left\Vert \left[ P_{n,j} P_{n',j} - P_{n+n',j} \right] \vert \psi \rangle \right\Vert_{1} \\
&= \left\Vert \left[ P_{n,j} P_{n',j} - P_{n+n',j} + \left( \left( A_{j}^{n+n'} \right)^{\dagger} - \left( A_{j}^{n\phantom{'}} \right)^{\dagger} \left( A_{j}^{n'} \right)^{\dagger} \right) + \left( \left( A_{j}^{n'} \right)^{\dagger} P_{n,j} - \left( A_{j}^{n'} \right)^{\dagger} P_{n,j} \right) \right] \vert \psi \rangle \right\Vert_{1} \\
&\leq \Vert P_{n,j} \Vert \left\Vert\left[ P_{n',j} - \left( A_{j}^{n'} \right)^{\dagger} \right] \vert \psi \rangle \right\Vert_{1} + \left\Vert \left( A_{j}^{n'} \right)^{\dagger} \right\Vert \left\Vert \left[ \left( A_{j}^{n\phantom{'}} \right)^{\dagger} - P_{n,j}\right] \vert \psi \rangle \right\Vert_{1} + \left\Vert \left[ P_{n+n',j} - \left( A_{j}^{n+n'} \right)^{\dagger} \right] \vert \psi \rangle \right\Vert_{1} \\
&\leq \sqrt{\epsilon} \left( \Vert P_{n,j} \Vert + 2 \right) \\
&\leq \sqrt{\epsilon} \left( \gamma(\nu) + 2 \right) \,,
\end{aligned}
\label{eq:pbounds}
\end{equation}
where we introduce $\gamma(\nu) \leq \Vert P_{n,j} \Vert$ for the following reason:
In the worst case, $\gamma(\nu) = d$ due to
\begin{equation}
\left\Vert P_{n,j} \right\Vert = \frac{1}{d} \left\Vert \sum_{k,s\in \fld} \omega^{\iota(j,k,s,\nu)} B_{k}^{n} \right\Vert \leq \frac{1}{d} \sum_{k,s\in \fld} \left\Vert B_{k}^{n} \right\Vert \leq d \,,
\end{equation}
for all $j \in \fld$.
But, we find that, when $\nu$ is a polynomial of degree three, $\mathrm{deg}(\nu) = 3$, it is $\left\Vert P_{n,j} \right\Vert \leq \sqrt{d}$ for all $j \in \fld$.
In this case, $\gamma(\nu) = \sqrt{d}$.
For every $j \in \fld$, Equation\,\eqref{eq:pbounds}, together with the identity~\eqref{eq:qfid}, is then
\begin{equation}
\left\Vert \sum_{ k,l \in \fld } g(j,k,n) g(j,l,n') \left( B_{k}^{n} B_{l}^{n'} - \omega^{2^{-1}nn'(k-l)} B_{(n+n')^{-1}(nk+n'l)}^{n+n'} \right) \vert \psi \rangle \right\Vert_{1} \leq \sqrt{\epsilon} \left( \gamma(\nu) + 2 \right) \,. \label{eq:fml-1}
\end{equation}
Now, we define the vector $R = (R_{k,l} )_{(k,l)\in\fld^{2}}$ with \[ R_{k,l} := B_{k}^{n} B_{l}^{n'} - \omega^{2^{-1}nn'(k-l)} B_{(n+n')^{-1}(nk+n'l)}^{n+n'} \,, \] and the matrix $G_{n} = (g(j,k,n))_{j,k\in\fld}$.
The matrix $G_{n}$ is unitary because $\sum_{j\in\fld} g(j,k,n) g(j,k',n)^{\ast} = \delta_{k,k'}$ and $\sum_{k\in\fld} g(j,k,n) g(j',k,n)^{\ast} = \delta_{j,j'}$.
There exists a basis $\{ T_{(k,l)} \}_{(k,l) \in \fld^{2}}$ such that
\begin{equation}
R_{(k,l)} = \sum_{p,q \in \fld} c_{(k,l)}^{(p,q)}\, T_{(p,q)} \,,
\end{equation}
for coefficients $c_{(k,l)}^{(p,q)}$.
Furthermore, we write 
\begin{equation}
b_{j}^{(p,q)} = \sum_{k,l \in \fld} g(j,k,n) g(j,l,n') \, c_{(k,l)}^{(p,q)} \,.
\end{equation}
Then, we convert the expression to the Hilbert-Schmidt norm
\begin{equation}
\sqrt{ \sum_{p,q \in \fld} \left\vert b_{j}^{(p,q)} \right\vert^2 } = \left\Vert \sum_{k,l \in \fld} b^{(p,q)}_{j} \right\Vert_{2} = \left\Vert \sum_{p,q \in \fld} b_{j}^{(p,q)} T_{(p,q)} \vert \psi \rangle \right\Vert_{2} \leq \sqrt{d} \left\Vert \sum_{p,q \in \fld} b_{j}^{(p,q)} T_{(p,q)} \vert \psi \rangle \right\Vert_{1} \leq \sqrt{d} \sqrt{\epsilon} \left( \gamma(\nu) + 2 \right) \,. \label{eq:fml-2}
\end{equation}
Then, we can bound the norm of the vector of coefficients $c^{(p,q)}_{(k,l)}$.
First,
\begin{equation}
\sqrt{ \sum_{k,l \in \fld} \left\vert c^{(p,q)}_{(k,l)} \right\vert^{2} } = \left\Vert \sum_{k,l \in \fld} c^{(p,q)}_{(k,l)} \right\Vert_{2} = \left\Vert \sum_{k,l \in \fld} ( G_{n} \otimes G_{n'} )^{-1} b^{(p,q)}_{(k,l)} \right\Vert_{2} = \left\Vert \sum_{k,l \in \fld} b^{(p,q)}_{j} \right\Vert_{2} \leq \sqrt{d} \sqrt{\epsilon} \left( \gamma(\nu) + 2 \right) \,,
\end{equation}
and then,
\begin{equation}
\left\vert c^{(p,q)}_{(k,l)} \right\vert \leq \sqrt{d} \sqrt{\epsilon} \left( \gamma(\nu) + 2 \right)
\end{equation}

Then, we bound
\begin{align}
\left\Vert \left( B_{k}^{n} B_{l}^{n'} - \omega^{2^{-1}nn'(k-l)} B_{(n+n')^{-1}(nk+n'l)}^{n+n'} \right) \vert \psi \rangle \right\Vert_{1} = \left\Vert R_{(k,l)} \vert \psi \rangle \right\Vert_{1} &= \Bigg\Vert \sum_{p,q \in \fld} c_{(k,l)}^{(p,q)}\, T_{(p,q)} \vert \psi \rangle \Bigg\Vert_{1} \\
&\leq \sum_{p,q \in \fld} \vert c_{(k,l)}^{(p,q)} \vert \left\Vert T_{(p,q)} \vert \psi \rangle \right\Vert_{1} \\
&\leq \sqrt{d} \sqrt{\epsilon} \left( \gamma(\nu) + 2 \right) \,,
\end{align}
Due to the symmetry of $A$ and $B$, it holds that
\begin{align}
\label{eq:finalboundA}\left\Vert \left( A_{k}^{n} A_{l}^{n'} - \omega^{2^{-1}nn'(k-l)} A_{(n+n')^{-1}(nk+n'l)}^{n+n'} \right) \vert \psi \rangle \right\Vert_{1} &\leq \sqrt{d} \sqrt{\epsilon} \left( \gamma(\nu) + 2 \right)\,,\\
\label{eq:finalboundB}\left\Vert \left( B_{k}^{n} B_{l}^{n'} - \omega^{2^{-1}nn'(k-l)} B_{(n+n')^{-1}(nk+n'l)}^{n+n'} \right) \vert \psi \rangle \right\Vert_{1} &\leq \sqrt{d} \sqrt{\epsilon} \left( \gamma(\nu) + 2 \right)\,.
\end{align}
For the robust $\omega$-commutation relations, we get
\begin{align}
&\left\Vert \left( B_{k}^{n} B_{l}^{n'} - \omega^{nn'(k-l)} B_{l}^{n'} B_{k}^{n} \right) \vert \psi \rangle \right\Vert_{1} \\[0.25cm]
&= \left\Vert \left( B_{k}^{n} B_{l}^{n'} - \omega^{2^{-1}nn'(k-l)} B_{(n+n')^{-1}(nk+n'l)}^{n+n'} + \omega^{2^{-1}nn'(k-l)} B_{(n+n')^{-1}(nk+n'l)}^{n+n'} - \omega^{nn'(k-l)} B_{l}^{n'} B_{k}^{n} \right) \vert \psi \rangle \right\Vert_{1} \\[0.25cm]
&\leq \left\Vert \left( B_{k}^{n} B_{l}^{n'} - \omega^{2^{-1}nn'(k-l)} B_{(n+n')^{-1}(nk+n'l)}^{n+n'} \right) \vert \psi \rangle \right\Vert_{1} + \left\Vert \left( \omega^{2^{-1}nn'(k-l)} B_{k}^{n} B_{l}^{n'} - \omega^{ nn'(k-l)} B_{(n+n')^{-1}(nk+n'l)}^{n+n'} \right) \vert \psi \rangle \right\Vert_{1} \\[0.25cm]
&\leq 2 \sqrt{d} \sqrt{\epsilon} \left( \gamma(\nu) + 2 \right)\,. 
\end{align}
Due to the symmetry of $A$ and $B$, it holds that
\begin{align}
\label{eq:finalboundcomA}\left\Vert \left( A_{k}^{n} A_{l}^{n'} - \omega^{nn'(k-l)} A_{l}^{n'} A_{k}^{n} \right) \vert \psi \rangle \right\Vert_{1} &\leq 2 \sqrt{d} \sqrt{\epsilon} \left( \gamma(\nu) + 2 \right)\,,\\
\label{eq:finalboundcomB}\left\Vert \left( B_{k}^{n} B_{l}^{n'} - \omega^{nn'(k-l)} B_{l}^{n'} B_{k}^{n} \right) \vert \psi \rangle \right\Vert_{1} &\leq 2 \sqrt{d} \sqrt{\epsilon} \left( \gamma(\nu) + 2 \right)\,.
\end{align}

\section{Robust self-testing relations - ($d = 3$)}
\label{sec:robust3}

Under the assumptions of Lemma~\ref{lem:com-ideal}, $ \langle \psi \vert 6 \mathds{1} - \mathcal{B}_{3} \vert \psi \rangle \leq \epsilon$ for an $\epsilon >0$.
The operator has a sum-of-positive-operators (SOPO) decomposition: $ 6 \mathds{1} - \mathcal{B}_{3} = C_{0}^{\dagger} C_{0}^{\phantom{\dagger}} + C_{1}^{\phantom{\dagger}} C_{1}^{\dagger} + C_{2}^{\phantom{\dagger}} C_{2}^{\dagger}$ with $C_{j}$ from Eq.\,\eqref{eq:sos3-j} for $j=0,1,2$.
If follows that $\langle \psi \vert C_{j}^{\dagger} C_{j}^{\phantom{\dagger}} \vert \psi \rangle \leq \epsilon$, and consequently,
\begin{equation}
\sqrt{ \langle \psi \vert C_{j}^{\dagger} C_{j}^{\phantom{\dagger}} \vert \psi \rangle } = \left\Vert C_{j} \vert \psi \rangle \right\Vert_{1} \leq \sqrt{\epsilon}
\label{eq:boundccdagger}
\end{equation}
for $j=0,1,2$, and the trace norm $\Vert \rho \Vert_{1} := \tr \sqrt{ \rho^{\dagger} \rho}$.
Form here onward, we follow the steps of the proof for the ideal case and bound the terms one after the other. 

\subsection{Bound $\Vert [\{ B_{0}^{2}, B_{1}^{2} \} + B_{2} ] \vert \psi \rangle \Vert_{1}$, $\Vert [\{ B_{0}, B_{1} \} + B_{2}^{2} ] \vert \psi \rangle \Vert_{1}$, and index permutations thereof}
The first step is to bound $\Vert [\{ B_{0}^{2}, B_{1}^{2} \} + B_{2} ] \vert \psi \rangle \Vert_{1}$ and index permutations thereof.
To this aim, $A_{0}^{2} = A_{0}^{\dagger}$ leads to
\begin{align}
\label{eq:cbound-0}&\phantom{\leq}~ \left\Vert \left[ \frac{e^{-2\ii\phi_{1}}}{3} \left( B_{0}^{2} + \omega B_{1}^{2} + B_{2}^{2}\right)^{2} - \frac{e^{\ii \phi_{1}}}{\sqrt{3}} \left( B_{0} + \omega^{2} B_{1} + B_{2}\right) \right] \vert \psi \rangle \right\Vert_{1} \\
\label{eq:cbound-1}&\leq \left\Vert \left[ \frac{e^{-2\ii\phi_{1}}}{3} \left( B_{0}^{2} + \omega B_{1}^{2} + B_{2}^{2}\right)^{2} - A_{0}^{\dagger} \right] \vert \psi \rangle \right\Vert_{1} + \left\Vert \left[ A_{0}^{\dagger} - \frac{e^{\ii \phi_{1}}}{\sqrt{3}} \left( B_{0} + \omega^{2} B_{1} + B_{2}\right) \right] \vert \psi \rangle \right\Vert_{1} \\
\label{eq:cbound-2}&\leq \left\Vert \left[ \frac{e^{2\ii\phi_{1}}}{3} \left( B_{0}^{2} + \omega B_{1}^{2} + B_{2}^{2}\right)^{2} - A_{0}^{2} \right] \vert \psi \rangle \right\Vert_{1} + \sqrt{\epsilon} \\
\label{eq:cbound-3}&\leq \left\Vert \left[ \frac{e^{-\ii\phi_{1}}}{\sqrt{3}} \left( B_{0}^{2} + \omega B_{1}^{2} + B_{2}^{2}\right) - A_{0} \right] \vert \psi \rangle \right\Vert_{1} \left\Vert \left[ \frac{e^{-\ii\phi_{1}}}{\sqrt{3}} \left( B_{0}^{2} + \omega B_{1}^{2} + B_{2}^{2}\right) + A_{0} \right] \vert \psi \rangle \right\Vert_{1} + \sqrt{\epsilon} \\
\label{eq:cbound-4}&\leq \sqrt{\epsilon} \, \left\Vert \left[ \frac{e^{-\ii\phi_{1}}}{\sqrt{3}} \left( B_{0}^{2} + \omega B_{1}^{2} + B_{2}^{2}\right) + A_{0} \right] \vert \psi \rangle \right\Vert_{1} + \sqrt{\epsilon}\\[5pt]
\label{eq:cbound-5}&\leq \frac{\sqrt{\epsilon}}{\sqrt{3}} \Big( \left\Vert B_{0}^{2} \right\Vert + \left\Vert B_{1}^{2} \right\Vert + \left\Vert B_{2}^{2} \right\Vert \Big) + \sqrt{\epsilon} \left\Vert A_{0} \right\Vert + \sqrt{\epsilon}\\[8pt]
\label{eq:cbound-6}&\leq \sqrt{\epsilon}\,(\sqrt{3} + 2) \,.
\end{align}
From Eq.\,\eqref{eq:cbound-0} to Eq.\,\eqref{eq:cbound-1}, add $A_{0}^{\dagger} - A_{0}^{\dagger}$ and use the triangle inequality.
From Eq.\,\eqref{eq:cbound-1} to Eq.\,\eqref{eq:cbound-2}, use Eq.\,\eqref{eq:boundccdagger}.
From Eq.\,\eqref{eq:cbound-2} to Eq.\,\eqref{eq:cbound-3}, use the inequality $\Vert D^{2} - E^{2} \Vert \leq \Vert D - E \Vert \Vert D + E \Vert$.
From Eq.\,\eqref{eq:cbound-3} to Eq.\,\eqref{eq:cbound-4}, use Eq.\,\eqref{eq:boundccdagger}.
From Eq.\,\eqref{eq:cbound-4} to Eq.\,\eqref{eq:cbound-5}, use the triangle inequality and $\left\Vert \left[ D \right] \vert \psi \rangle \right\Vert_{1} \leq \left\Vert D \right\Vert$.
From Eq.\,\eqref{eq:cbound-5} to Eq.\,\eqref{eq:cbound-6}, all operators are unitary such that their norm is $\Vert A_{j} \Vert = \Vert B_{j} \Vert =1$ for $j = 0,1,2$.

One can show the same for $A_{1}^{2} = A_{1}^{\dagger}$ and $A_{2}^{2} = A_{2}^{\dagger}$ where the reasoning is equivalent up to permutation of indices such that
\begin{align}
\label{eq:cboundall-1} \left\Vert \left[ \frac{e^{-2\ii\phi_{1}}}{3} \left( B_{0}^{2} + \omega B_{1}^{2} + B_{2}^{2}\right)^{2} - \frac{e^{\ii \phi_{1}}}{\sqrt{3}} \left( B_{0} + \omega^{2} B_{1} + B_{2}\right) \right] \vert \psi \rangle \right\Vert_{1} & \leq \sqrt{\epsilon}(\sqrt{3}+2)\\
\label{eq:cboundall-2} \left\Vert \left[ \frac{e^{-2\ii\phi_{1}}}{3} \left( \omega B_{0}^{2} + B_{1}^{2} + B_{2}^{2}\right)^{2} - \frac{e^{\ii \phi_{1}}}{\sqrt{3}} \left(\omega^{2} B_{0} + B_{1} + B_{2}\right) \right] \vert \psi \rangle \right\Vert_{1} & \leq \sqrt{\epsilon}(\sqrt{3}+2)\\
\label{eq:cboundall-3} \left\Vert \left[ \frac{e^{-2\ii\phi_{1}}}{3} \left( B_{0}^{2} + B_{1}^{2} + \omega B_{2}^{2}\right)^{2} - \frac{e^{\ii \phi_{1}}}{\sqrt{3}} \left( B_{0} + B_{1} + \omega^{2} B_{2}\right) \right] \vert \psi \rangle \right\Vert_{1} & \leq \sqrt{\epsilon}(\sqrt{3}+2)
\end{align}
We evaluate the left-hand side of Eqs.\,\eqref{eq:cboundall-1}-\eqref{eq:cboundall-3} as for the ideal case using $1-\sqrt{3}e^{3\ii \phi_{1}} = \omega$, to arrive at $\left\Vert F_{j} \vert \psi \rangle \right\Vert_{1} \leq 3 \sqrt{\epsilon}(\sqrt{3}+2)$ for $j \in \{ 0,1,2 \}$ with
\begin{align}
\label{eq:bbboundall-1} F_{0} := \omega \{ B_{0}^{2},B_{1}^{2}\} +\phantom{\omega} \{ B_{0}^{2},B_{2}^{2} \} + \omega \{ B_{1}^{2},B_{2}^{2}\} + \omega B_{0} + \phantom{\omega}B_{1} + \omega B_{2} \,, \\
\label{eq:bbboundall-2} F_{1} := \omega \{ B_{0}^{2},B_{1}^{2}\} + \omega \{ B_{0}^{2},B_{2}^{2} \} + \phantom{\omega} \{ B_{1}^{2},B_{2}^{2}\} + \phantom{\omega}B_{0} + \omega B_{1} + \omega B_{2} \,,\\
\label{eq:bbboundall-3} F_{2} := \phantom{\omega} \{ B_{0}^{2},B_{1}^{2}\} + \omega \{ B_{0}^{2},B_{2}^{2} \} + \omega \{ B_{1}^{2},B_{2}^{2}\} + \omega B_{0} + \omega B_{1} + \phantom{\omega}B_{2} \,.
\end{align}
With these,
\begin{align}
\label{eq:boundalgebraB-1}\left\Vert \left[ \{ B_{0}^{2}, B_{1}^{2} \} + B_{2} \right] \vert \psi \rangle \right\Vert_{1} = \frac{1}{3} \left\Vert \left[ F_{0} + F_{1} + \omega F_{2} \right] \vert \psi \rangle \right\Vert_{1} \leq \frac{1}{3} \left( \left\Vert F_{0} \vert \psi \rangle \right\Vert_{1} + \left\Vert F_{1} \vert \psi \rangle \right\Vert_{1} + \left\Vert F_{2} \vert \psi \rangle \right\Vert_{1} \right) &\leq 3\sqrt{\epsilon}(\sqrt{3}+2) \,,\\
\left\Vert \left[ \{ B_{0}^{2}, B_{2}^{2} \} + B_{1} \right] \vert \psi \rangle \right\Vert_{1} = \frac{1}{3} \left\Vert \left[ \omega F_{0} + F_{1} + F_{2} \right] \vert \psi \rangle \right\Vert_{1} \leq \frac{1}{3} \left(\left\Vert F_{0} \vert \psi \rangle \right\Vert_{1} + \left\Vert F_{1} \vert \psi \rangle \right\Vert_{1} + \left\Vert F_{2} \vert \psi \rangle \right\Vert_{1} \right) &\leq \label{eq:boundalgebraB-2}3\sqrt{\epsilon}(\sqrt{3}+2) \,,\\
\label{eq:boundalgebraB-3}\left\Vert \left[ \{ B_{1}^{2}, B_{2}^{2} \} + B_{0} \right] \vert \psi \rangle \right\Vert_{1} = \frac{1}{3} \left\Vert \left[ F_{0} + \omega F_{1} + F_{2} \right] \vert \psi \rangle \right\Vert_{1} \leq \frac{1}{3} \left(\left\Vert F_{0} \vert \psi \rangle \right\Vert_{1} + \left\Vert F_{1} \vert \psi \rangle \right\Vert_{1} + \left\Vert F_{2} \vert \psi \rangle \right\Vert_{1} \right) &\leq 3\sqrt{\epsilon}(\sqrt{3}+2) \,.
\end{align}
The same bounds hold for the complex conjugate of the expressions, which can be achieved by starting from $A_{j} = A^{\dagger}_{j} A^{\dagger}_{j}$ for all $j \in \{ 0,1,2 \}$.
\begin{align}
\label{eq:boundalgebraB-1*} \left\Vert \left[ \{ B_{0}, B_{1} \} + B_{2}^{2} \right] \vert \psi \rangle \right\Vert_{1} &\leq 3\sqrt{\epsilon}(\sqrt{3}+2) \,,\\
\left\Vert \left[ \{ B_{0}, B_{2} \} + B_{1}^{2} \right] \vert \psi \rangle \right\Vert_{1} &\leq \label{eq:boundalgebraB-2*}3\sqrt{\epsilon}(\sqrt{3}+2) \,,\\
\label{eq:boundalgebraB-3*}\left\Vert \left[ \{ B_{1}, B_{2} \} + B_{0}^{2} \right] \vert \psi \rangle \right\Vert_{1} &\leq 3\sqrt{\epsilon}(\sqrt{3}+2) \,.
\end{align}

All other steps follow in the exact same manner.

For later use, we derive the bounds
\begin{align}
\left\Vert \left[ B_{0} B_{1} B_{2} - B_{1} B_{2} B_{0} \right] \vert \psi \rangle \right\Vert_{1}
&\leq \left\Vert \left[ B_{0} B_{1} B_{2} + B_{1} B_{0} B_{2} - B_{1} B_{0} B_{2} - B_{1} B_{2} B_{0} \right] \vert \psi \rangle \right\Vert_{1} \\
\nonumber &= \left\Vert \left[ \{ B_{0}, B_{1}\} B_{2} - B_{1} \{ B_{0}, B_{2}\} \right] \vert \psi \rangle \right\Vert_{1}\\
\nonumber &= \left\Vert \left[ \{ B_{0}, B_{1}\} B_{2} - B_{1} \{ B_{0}, B_{2}\} + (B_{2}^3 - B_{1}^3) \right] \vert \psi \rangle \right\Vert_{1}\\
\label{eq:cyclepermutations-1}&\leq 6\sqrt{\epsilon}(\sqrt{3}+2)\,, \\
\label{eq:cyclepermutations-2}\left\Vert \left[ B_{0} B_{1} B_{2} - B_{2} B_{0} B_{1} \right] \vert \psi \rangle \right\Vert_{1} &\leq 6\sqrt{\epsilon}(\sqrt{3}+2) \,, \\
\label{eq:cyclepermutations-3} \left\Vert \left[ B_{1} B_{2} B_{0} - B_{2} B_{0} B_{1} \right] \vert \psi \rangle \right\Vert_{1} &\leq 6\sqrt{\epsilon}(\sqrt{3}+2) \,, \\[0.25cm]
\label{eq:cyclepermutations-4}\left\Vert \left[ B_{0} B_{2} B_{1} - B_{2} B_{1} B_{0} \right] \vert \psi \rangle \right\Vert_{1} &\leq 6\sqrt{\epsilon}(\sqrt{3}+2)\,, \\
\label{eq:cyclepermutations-5}\left\Vert \left[ B_{0} B_{2} B_{1} - B_{1} B_{0} B_{2} \right] \vert \psi \rangle \right\Vert_{1} &\leq 6\sqrt{\epsilon}(\sqrt{3}+2) \,, \\
\label{eq:cyclepermutations-6} \left\Vert \left[ B_{2} B_{1} B_{0} - B_{1} B_{0} B_{2} \right] \vert \psi \rangle \right\Vert_{1} &\leq 6\sqrt{\epsilon}(\sqrt{3}+2) \,.
\end{align}

\subsection{Commutation elements}
We consider the three commutation elements
\begin{align}
Q &= B_{0} B_{1} B_{0}^{2} B_{1}^{2}\,,\\
Q^{\prime} &= B_{2} B_{0} B_{2}^{2} B_{0}^{2}\,,\\
Q^{\prime\prime} &= B_{1} B_{2} B_{1}^{2} B_{2}^{2}\,,
\end{align}
It is useful to write the commutation elements $Q,Q^{\prime},Q^{\prime\prime}$ as
\begin{align}
Q &= B_{0} B_{1} ( \{ B_{0}^{2},B_{1}^{2} \} + B_{2} ) - \mathds{1} - B_{0} B_{1} B_{2} \,, \\
Q^{\prime} &= B_{2} B_{0} ( \{ B_{2}^{2},B_{0}^{2} \} + B_{1} ) - \mathds{1} - B_{2} B_{0} B_{1} \,, \\
Q^{\prime\prime} &= B_{1} B_{2} ( \{ B_{1}^{2},B_{2}^{2} \} + B_{0} ) - \mathds{1} - B_{1} B_{2} B_{0} \,.
\end{align}
Moreover, note that
\begin{align}
\left\Vert \left[ Q + Q^{\dagger} + \mathds{1} \right] \vert \psi \rangle \right\Vert_{1} &= \left\Vert \left[ \{ B_{0}, B_{1}\} \{ B_{0}^{2},B_{1}^{2} \} - \mathds{1} \right] \vert \psi \rangle \right\Vert_{1} \\
&= \left\Vert \left[ \{ B_{0}, B_{1}\} \{ B_{0}^{2},B_{1}^{2} \} + \{ B_{0}, B_{1}\} B_2 - \{ B_{0}, B_{1}\} B_2 - \mathds{1} \right] \vert \psi \rangle \right\Vert_{1} \\
&\leq 6 \sqrt{\epsilon}(\sqrt{3}+2)\,. \label{eq:Q3rootunity}
\end{align}
In the following, we show that $Q$ is `almost' central:
First, bound
\begin{align}
\left\Vert \left[ [B_{0},Q] \right] \vert \psi \rangle \right\Vert_{1} &= \left\Vert \left[ [ B_{0}, B_0 B_1 (\{B_0^2, B_1^2\} + B_2) - B_0 B_1 B_2 - \mathds{1} ] \right] \vert \psi \rangle \right\Vert_{1} \\
\nonumber&= \Vert [ B_{0}^{2} B_{1} ( \{B_{0}^{2}, B_{1}^{2}\} + B_{2} ) - B_{0}^{2} B_{1} B_{2} \\
\label{eq:boundqBs-1}&\phantom{=}~- B_{0} B_{1} ( \{B_{0}^{2}, B_{1}^{2}\} + B_{2} ) B_{0} + B_{0} B_{1} B_{2} B_{0} ] \vert \psi \rangle \Vert_{1} \\
\label{eq:boundqBs-2}&\leq 6\sqrt{\epsilon}(\sqrt{3}+2) + \Vert [ B_{0}^{2} B_{1} B_{2} - B_{0} B_{1} B_{2} B_{0} ] \vert \psi \rangle \Vert_{1} \\
\label{eq:boundqBs-3}&= 12\sqrt{\epsilon}(\sqrt{3}+2) \\
\label{eq:boundqB2-1} \left\Vert [B_{1},Q] \vert \psi \rangle \right\Vert_{1} &\leq 12\sqrt{\epsilon}(\sqrt{3}+2) \\
\left\Vert [B_{2},Q] \vert \psi \rangle \right\Vert_{1} &\leq 12\sqrt{\epsilon}(\sqrt{3}+2)
\end{align}
In Eq.\,\eqref{eq:boundqBs-1}, we inflated the term, then used\,\eqref{eq:boundalgebraB-1} in Eq.\,\eqref{eq:boundalgebraB-2}, and in Eq.\,\eqref{eq:boundalgebraB-3} the final bound comes from Eq.\,\eqref{eq:cyclepermutations-1}.
Write $S := Q - \frac{1}{3}\tr(Q) \mathds{1}$, then $2\Vert S \Vert \geq \Vert [U,S] \Vert = \Vert [U,Q] \Vert$ for $\Vert U \Vert$.
As a result, $\Vert S \Vert = \frac{1}{2} \sup_{\Vert U \Vert = 1} \Vert [U,Q] \Vert$, and
\begin{align}
\left\Vert \left[ Q - \lambda \mathds{1} \right] \vert \psi \rangle \right\Vert_{1} \leq \left\Vert Q - \lambda \mathds{1} \right\Vert = \frac{1}{2} \sup_{\Vert U \Vert =1 } \left\Vert [Q,U] \right\Vert \leq 6\sqrt{\epsilon} (\sqrt{3}+2) \label{eq:qlambda1bound}
\end{align}
Since $Q$ is unitary, thus normal, we can apply the spectral mapping theorem 
\begin{equation}
\left\Vert Q - \lambda \mathds{1} \right\Vert = \max_{\mu \in \sigma(Q)} \vert \mu - \lambda \vert \leq 6\sqrt{\epsilon} (\sqrt{3}+2) \,.\label{eq:qevall}
\end{equation}
\subsection{Bounding $\Vert Q - \omega^k \mathds{1} \Vert$}
From Eq.\,\eqref{eq:Q3rootunity},
\begin{align}
\vert \lambda + \lambda^{\ast} + 1 \vert &= \left\Vert \left[(\lambda\mathds{1}- Q) + (\lambda^{\ast}\mathds{1}- Q^{\dagger}) + (\mathds{1} + Q + Q^{\dagger}) \right] \vert \psi \rangle \right\Vert_{1} \\
&\leq 18\sqrt{\epsilon} (\sqrt{3}+2) \,.
\end{align}
Write $\lambda = e^{\ii \theta}$, then $ \lambda + \lambda^{\ast} + 1 = 2 \cos{\theta} + 1 $.
Assuming that $\epsilon' := 18\sqrt{\epsilon} (\sqrt{3}+2) < 1/2 < \vert \omega - \omega^{2}\vert$, confines $\theta$ to a $\epsilon'$-neighborhood of $\theta = 2\pi/3$ or $\theta = 4\pi/3$, such that $\vert \lambda - \omega^{k} \vert \leq \epsilon'$ for $k = 1,2$.

More formally, we show, first, that the eigenvalues of $Q$ are `almost' $\omega^{k}$ for $k =1$ or $k=2$, starting with
\begin{align}
\left\Vert \left[ Q^3 - \mathds{1} \right] \vert \psi \rangle \right\Vert_{1} &= \left\Vert \left[ Q^{2} B_{0} B_{1} B_{0}^{2} B_{1}^{2} - Q B_{0} Q B_{1} B_{0}^{2} B_{1}^{2} + Q B_{0} Q B_{1} B_{0}^{2} B_{1}^{2} - \mathds{1} \right] \vert \psi \rangle \right\Vert_{1} \\
&\leq 12\sqrt{\epsilon}(\sqrt{3}+2) + \left\Vert \left[ Q B_{0} Q B_{1} B_{0}^{2} B_{1}^{2} - \mathds{1} \right] \vert \psi \rangle \right\Vert_{1}\\
&= 12\sqrt{\epsilon}(\sqrt{3}+2) + \left\Vert \left[ Q B_{0}^{2} B_{1} B_{0} B_{1}^{2} - \mathds{1} \right] \vert \psi \rangle \right\Vert_{1}\\
&\leq 12\sqrt{\epsilon}(\sqrt{3}+2) + \left\Vert \left[ Q B_{0}^{2} B_{1} B_{0} B_{1}^{2} - B_{0} Q B_{0} B_{1} B_{0} B_{1}^{2} + B_{0} Q B_{0} B_{1} B_{0} B_{1}^{2} - \mathds{1} \right] \vert \psi \rangle \right\Vert_{1}\\
&\leq 24\sqrt{\epsilon}(\sqrt{3}+2) + \left\Vert \left[ B_{0} Q B_{0} B_{1} B_{0} B_{1}^{2} - B_{0}^{2} Q B_{1} B_{0} B_{1}^{2} + B_{0}^{2} Q B_{1} B_{0} B_{1}^{2} - \mathds{1} \right] \vert \psi \rangle \right\Vert_{1}\\
&\leq 36\sqrt{\epsilon}(\sqrt{3}+2) + \left\Vert \left[ B_{0}^{2} Q B_{1} B_{0} B_{1}^{2} - \mathds{1} \right] \vert \psi \rangle \right\Vert_{1}\\
&= 36\sqrt{\epsilon}(\sqrt{3}+2) + \left\Vert \left[ B_{0}^{2} B_{0} B_{1} B_{0}^{2} B_{1}^{2} B_{1} B_{0} B_{1}^{2} - \mathds{1} \right] \vert \psi \rangle \right\Vert_{1}\\
&= 36\sqrt{\epsilon}(\sqrt{3}+2)
\end{align}
As a result, the eigenvalues of $Q$ are close to the third roots of unity $\{ 1, \omega, \omega^{2} \}$.
In particular, the spectral mapping theorem of normals (applies since the $B_{j}$ are unitary) gives
\begin{equation}
\Vert Q^3 - 1 \Vert = \max_{\lambda \in \sigma(Q)} \vert \lambda^3 -1 \vert \leq 36\sqrt{\epsilon}(\sqrt{3}+2) \,.
\end{equation}
In the following, we define $\epsilon^{\prime} = 36\sqrt{\epsilon}(\sqrt{3}+2)$.
From Eq.\,\eqref{eq:qevall}, every eigenvalue $\lambda \in \sigma(Q)$ lies in the union of the three $\epsilon^{\prime}$-disks $D_{\epsilon^{\prime}}(\omega^{k})$ for $k=0,1,2$.
By assuming $\epsilon^{\prime} < 1/2$, we can make the three disk disjoint since $\vert \omega^k - \omega^j \vert = \sqrt{3}$ for $k \neq j$.
Then, it is
\begin{equation}
\vert \lambda - \omega^k \vert \leq \epsilon^{\prime} < \frac{1}{2} \,,
\end{equation}
for all $\lambda$, which implies that all eigenvalues lie in the same disk, i.e.~$k = k^{\ast}$ is fixed.
Finally,
\begin{align}
\Vert Q- \omega^{k^{\ast}} \mathds{1} \Vert &= \max_{\lambda \in \sigma(Q)} \vert \lambda - \omega^{k^{\ast}} \vert \leq \epsilon^{\prime}
\end{align}

Lastly, we show that $k^{\ast} \neq 0$ by contradiction.
So assume $\Vert Q -\mathds{1} \Vert \leq \epsilon^{\prime}$, then $\Vert [B_{0},B_{1}] \Vert = \Vert Q -\mathds{1} \Vert \leq \epsilon^{\prime}$.
For the anticommutator, we use the inequality, given that $U$ and $V$ are unitary, therefore $UV$ is unitary and $\Vert UV \Vert = 1$,
\begin{equation}
\Vert \{ U,V \} \Vert = \Vert UV + 2VU - VU \Vert \geq \vert \Vert [U,V] \Vert - 2 \vert\,.
\end{equation}
Then,
\begin{equation}
\Vert \{B_{0},B_{1} \} \Vert \geq 2 - \Vert [B_{0},B_{1}] \Vert \,,
\end{equation}
and, furthermore,
\begin{equation}
\Vert \{ B_{0}, B_{1} \} + B_{2}^{2} \Vert \geq \vert \Vert \{ B_{0}, B_{1} \} \Vert - \Vert B_{2}^{2} \Vert \vert = \Vert \{ B_{0}, B_{1} \} \Vert - 1 \geq 1 - \epsilon^{\prime} \,.
\end{equation}
This, however, contradicts $\Vert \{ B_{0}, B_{1} \} + B_{2}^{2} \Vert \leq 3\sqrt{\epsilon}(\sqrt{3}+2) $ from Eq.\,\eqref{eq:boundalgebraB-1} for a small enough $\epsilon^{\prime}$.

\subsection{Commutation elements $Q^{\prime}$ and $Q^{\prime\prime}$}

All inequalities and arguments also hold for $Q^{\prime}$ and $Q^{\prime\prime}$ and cyclic permutations of $B_{0},B_{1},B_{2}$.
As a final step, we need to show that $k^{\ast}$ is the same for $Q$, $Q^{\prime}$, and $Q^{\prime\prime}$.
First note that $Q,Q^{\prime},Q^{\prime\prime}$ are mutually close:
\begin{align}
\label{eq:qqclose-1}\left\Vert \left[ Q^{\phantom{\prime}} - Q^{\prime} \right] \vert \psi \rangle \right\Vert_{1} &= \left\Vert \left[ B_{0} B_{1} B_{0}^{2} B_{1}^{2} - B_{2} B_{0} B_{2}^{2} B_{0}^{2} \right] \vert \psi \rangle \right\Vert_{1} \\
\nonumber &= \big\Vert \big[ B_{0} B_{1} ( \{ B_{0}^{2}, B_{1}^{2}\} + B_{2} ) - B_{0} B_{1} B_{2} - \mathds{1} \\ 
\label{eq:qqclose-2}&\phantom{=\big\Vert \big[} - B_{2} B_{0} ( \{ B_{2}^{2}, B_{0}^{2}\} + B_{1} ) + B_{2} B_{0} B_{1} + \mathds{1} \big] \vert \psi \rangle \big\Vert_{1} \\
\label{eq:qqclose-3}&\leq 9 \sqrt{\epsilon} (\sqrt{3} + 2) \,,\\
\left\Vert \left[ Q^{\phantom{\prime}} - Q^{\prime\prime} \right] \vert \psi \rangle \right\Vert_{1} &\leq 9 \sqrt{\epsilon} (\sqrt{3} + 2) \,, \\
\left\Vert \left[ Q^{\prime} - Q^{\prime\prime} \right] \vert \psi \rangle \right\Vert_{1} &\leq 9 \sqrt{\epsilon} (\sqrt{3} + 2) \,.
\end{align}
From Eq.\,\eqref{eq:qqclose-1} to Eq.\,\eqref{eq:qqclose-2}, we inflate the expression.
From Eq.\,\eqref{eq:qqclose-2} to Eq.\,\eqref{eq:qqclose-3}, we use the triangle inequality and the bounds from Eq.\,\eqref{eq:boundalgebraB-1}, Eq.\,\eqref{eq:boundalgebraB-2}, and Eq.\,\eqref{eq:cyclepermutations-1}.
As a result, the eigenvalues of the commutation elements, fixed by $k^{\ast}$, are mutually `almost' the same:
\begin{align}
\vert \omega^{k^{\ast}} - \omega^{k^{\prime\ast}} \vert &= \left\Vert \left[ \omega^{k^{\ast}}\mathds{1} - \omega^{k^{\prime\ast}} \mathds{1} + Q - Q + Q^{\prime} - Q^{\prime} \right] \vert \psi \rangle \right\Vert_{1} \\
&\leq \left\Vert \left[ \omega^{k^{\ast}}\mathds{1} - Q \right] \vert \psi \rangle \right\Vert_{1} + \left\Vert \left[ Q^{\prime} - \omega^{k^{\prime\ast}} \mathds{1} \right] \vert \psi \rangle \right\Vert_{1} + \left\Vert \left[ Q - Q^{\prime} \right] \vert \psi \rangle \right\Vert_{1} \\
&\leq 2\epsilon^{\prime} + 9 \sqrt{\epsilon} (\sqrt{3} + 2) \\
\vert \omega^{k^{\ast}} - \omega^{k^{\prime\prime\ast}} \vert &\leq 2\epsilon^{\prime} + 9 \sqrt{\epsilon} (\sqrt{3} + 2) \\
\vert \omega^{k^{\prime\ast}} - \omega^{k^{\prime\prime\ast}} \vert &\leq 2\epsilon^{\prime} + 9 \sqrt{\epsilon} (\sqrt{3} + 2)
\end{align}
Since $2\epsilon^{\prime} + 9 \sqrt{\epsilon} (\sqrt{3} + 2) = ( 2 + \tfrac{1}{4} ) \epsilon^{\prime} < 1 + \tfrac{1}{8} < \sqrt{3}$, which is the distance between any $\omega^j$ and $\omega^k$ for $j\neq k$, we conclude that $k^{\ast} =k^{\prime\ast} =k^{\prime\prime\ast} \in \{ 1,2 \}$.

\subsection{Twisted commutation relations}
For a small enough $\epsilon$, i.e.~$36\sqrt{\epsilon}(\sqrt{3}+2) < 1/2$, it follows from Eq.\,\eqref{eq:qevall} and its following considerations that 
\begin{align}
\left\Vert [B_{0}B_{1} - \omega^{k^{\ast}} B_{1} B_{0} ] \vert \psi \rangle \right\Vert_{1} &\leq 9\sqrt{\epsilon} (\sqrt{3}+2) \,, \label{eq:finalbound3-1} \\
\left\Vert [B_{2}B_{0} - \omega^{k^{\ast}} B_{0} B_{2} ] \vert \psi \rangle \right\Vert_{1} &\leq 9\sqrt{\epsilon} (\sqrt{3}+2) \,, \label{eq:finalbound3-2} \\
\left\Vert [B_{1}B_{2} - \omega^{k^{\ast}} B_{2} B_{1} ] \vert \psi \rangle \right\Vert_{1} &\leq 9\sqrt{\epsilon} (\sqrt{3}+2) \,, \label{eq:finalbound3-3}
\end{align}
for a $k^{\ast} \in \{ 1,2 \}$.
The choice of $k^{\ast}$ remains, as for the ideal case $\epsilon=0$, a degree of freedom.
Here, we choose $k^{\ast} = 1$.

\section{Robust isometry}

To bound the isometry, we require two bounds: 
First, the bounds on the twisted commutation relations from Eq.\,\eqref{eq:finalboundA} and Eq.\,\eqref{eq:finalboundB}, for $d >3$,
\begin{align}
\left\Vert \left( A_{k}^{n} A_{l}^{n'} - \omega^{2^{-1}nn'(k-l)} A_{(n+n')^{-1}(nk+n'l)}^{n+n'} \right) \vert \psi \rangle \right\Vert_{1} &\leq \sqrt{d} \sqrt{\epsilon} \left( \gamma(\nu) + 2 \right) =: \sqrt{\epsilon} \, \mu_{d}\,,\\
\left\Vert \left( B_{k}^{n} B_{l}^{n'} - \omega^{2^{-1}nn'(k-l)} B_{(n+n')^{-1}(nk+n'l)}^{n+n'} \right) \vert \psi \rangle \right\Vert_{1} &\leq \sqrt{d} \sqrt{\epsilon} \left( \gamma(\nu) + 2 \right) =: \sqrt{\epsilon} \, \mu_{d}\,,
\end{align}
and, for $d=3$, those from Eqs.\,\eqref{eq:finalbound3-1}\,-\,\eqref{eq:finalbound3-3}:
\begin{equation}
\left\Vert [B_{i}B_{j} - \omega^{j-i} B_{j} B_{i} ] \vert \psi \rangle \right\Vert_{1} \leq 9\sqrt{\epsilon} (\sqrt{3}+2) =: \sqrt{\epsilon}\,\mu_{3}\,.
\end{equation}
Second, the bound on the relations between the $A_{j}$ and $B_{j}$ for all $j\in \fld$:
For $d>3$, Eq.\,\eqref{eq:boundccdagger5} and Eq.\,\eqref{eq:sos-j} give
\begin{equation}
\left\Vert \left[ A_{j}^{n} - \sum_{ k \in \fld } g(j,k,n)^{\ast} \left( B_{k}^{n} \right)^{\dagger} \right] \vert \psi \rangle \right\Vert_{1} \leq \sqrt{\epsilon} \,.
\end{equation}
For $d=3$, Eq.\,\eqref{eq:boundccdagger} and Eq.\,\eqref{eq:sos3-j} give
\begin{align}
\left\Vert \left[A_{0} - \frac{1}{\sqrt{3}}\left( e^{-\ii \phi_{1}} B_{0}^{2} + e^{-\ii \phi_{2}} B_{1}^{2} + e^{-\ii \phi_{1}} B_{2}^{2} \right) \right] \vert \psi \rangle \right\Vert_{1} &\leq \sqrt{\epsilon} \,, \\
\left\Vert \left[A_{1} - \frac{1}{\sqrt{3}}\left( e^{-\ii \phi_{2}} B_{0}^{2} + e^{-\ii \phi_{1}} B_{1}^{2} + e^{-\ii \phi_{1}} B_{2}^{2} \right) \right] \vert \psi \rangle \right\Vert_{1} &\leq \sqrt{\epsilon} \,, \\
\left\Vert \left[A_{2} - \frac{1}{\sqrt{3}}\left( e^{-\ii \phi_{1}} B_{0}^{2} + e^{-\ii \phi_{1}} B_{1}^{2} + e^{-\ii \phi_{2}} B_{2}^{2} \right) \right] \vert \psi \rangle \right\Vert_{1} &\leq \sqrt{\epsilon} \,.
\end{align}

Then, we go through Eqs.\,\eqref{eq:isoeval-1}\,-\,\eqref{eq:isoeval-10} and bound them using the triangle inequality:
\begin{align}
\label{eq:isoevalrob-1} &\left\Vert ( V_{A} \otimes V_{B} ) \vert \psi \rangle - \vert \psi \rangle \otimes \vert aux \rangle \right\Vert_{1} \\
\label{eq:isoevalrob-2} &= \left\Vert \frac{1}{d^{2}} \sum_{n,m \in \fld} \vert n,m\rangle \otimes \Big( \left( A_{0}^{n} \right)^{\dagger} \sum_{j \in \fld} \omega^{-jn} (\omega^{-2^{-1}j} A_{0}^{\dagger} A_{j}) \Big) \Bar{B}_{m} \vert \psi \rangle - \vert \psi \rangle \otimes \vert aux \rangle \right\Vert_{1} \\
\label{eq:isoevalrob-3} &\leq \sqrt{\epsilon} \, d(d-1) \mu_{d} + \left\Vert \frac{1}{d^{2}} \sum_{n,m,j \in \fld} \vert n,m\rangle \otimes \omega^{-2^{-1}jn} A_{-n^{-1}j}^{-n} \Bar{B}_{m} \vert \psi \rangle - \vert \psi \rangle \otimes \vert aux \rangle \right\Vert_{1} \\
\label{eq:isoevalrob-4} &\leq \sqrt{\epsilon} \, d(d-1) (\mu_{d} + 1) + \left\Vert \frac{1}{d^{3}} \sum_{n,m,j,l,s \in \fld} \vert n,m\rangle \otimes \Bar{B}_{m} \, \omega^{-2^{-1}jn} \omega^{ (j-l) s + \nu_{s+ 2^{-1}n} - \nu_{s- 2^{-1}n} } B_{n^{-1}l}^{n} \vert \psi \rangle - \vert \psi \rangle \otimes \vert aux \rangle \right\Vert_{1} \\
\label{eq:isoevalrob-5} &\leq \sqrt{\epsilon} \, d(d-1) (2\mu_{d} + 1) + \left\Vert \frac{1}{d^{2}} \sum_{n,r,k,l \in \fld} \omega^{ \nu_{n} - \nu_{0} } \vert n, n+r \rangle \otimes \omega^{-kr} B_{0}^{-r} \omega^{-2^{-1}(k+l)n} B_{-n^{-1}k}^{-n} B_{n^{-1}l}^{n} \vert \psi \rangle - \vert \psi \rangle \otimes \vert aux \rangle \right\Vert_{1} \\
\label{eq:isoevalrob-6} &\leq \sqrt{\epsilon} \, d(d-1) (4\mu_{d} + 1) + \left\Vert \frac{1}{d^{2}} \sum_{n,r,k,l \in \fld} \omega^{ \nu_{n} - \nu_{0} } \vert n,n+r \rangle \otimes \omega^{-kr} B_{0}^{-r} \omega^{- 2^{-1}(l+k)} B_{0}^{\dagger} B_{k+l} \vert \psi \rangle - \vert \psi \rangle \otimes \vert aux \rangle \right\Vert_{1} \\
\label{eq:isoevalrob-7} &= \sqrt{\epsilon} \, d(d-1) (4\mu_{d} + 1) + \left\Vert \frac{1}{d} \sum_{n,t \in \fld} \omega^{ \nu_{n} - \nu_{0} } \vert n,n \rangle \otimes \omega^{- 2^{-1}t} B_{0}^{\dagger} B_{t} \vert \psi \rangle - \vert \psi \rangle \otimes \vert aux \rangle \right\Vert_{1} \\
\label{eq:isoevalrob-8} &= \sqrt{\epsilon} \, d(d-1) (4\mu_{d} + 1) \,.
\end{align}
For the bound on $\left\Vert ( V_{A} \otimes V_{B} ) A_{u}^{v} \vert \psi \rangle - \left(T_{(v,vu)} \otimes \mathds{1}\right) \vert \psi \rangle \otimes \vert aux \rangle \right\Vert_{1}$, $\left\Vert ( V_{A} \otimes V_{B} ) B_{u}^{v} \vert \psi \rangle - \left( \mathds{1} \otimes T_{(v,vu)} \right) \vert \psi \rangle \otimes \vert aux \rangle \right\Vert_{1}$ , we apply the same inequalities as above but, since we apply the $\omega$-relation~\eqref{eq:finalboundA} and~\eqref{eq:finalboundB} twice, it is
\begin{align}
\left\Vert ( V_{A} \otimes V_{B} ) A_{u}^{v} \vert \psi \rangle - \left(T_{(v,vu)} \otimes \mathds{1}\right) \vert \psi \rangle \otimes \vert aux \rangle \right\Vert_{1} &\leq \sqrt{\epsilon} \, d(d-1) (\mu_{d}(4+1/d) + 1) \,, \\
\left\Vert ( V_{A} \otimes V_{B} ) B_{u}^{v} \vert \psi \rangle - \left( \mathds{1} \otimes T_{(v,vu)} \right) \vert \psi \rangle \otimes \vert aux \rangle \right\Vert_{1} &\leq \sqrt{\epsilon} \, d(d-1) (\mu_{d}(4+1/d) + 1) \,.
\end{align}

\end{document}